%% file: stationaryEMD.tex
\documentclass[a4paper,11pt]{amsart}

\usepackage{url}
\usepackage{graphicx}
\usepackage{amssymb}
\usepackage{color}
\usepackage[normalem]{ulem}
\usepackage[utf8]{inputenc}

{\catcode `\@=11 \global\let\AddToReset=\@addtoreset}
\AddToReset{equation}{section}

\usepackage{cite}

\newcommand{\mcD}{{\mycal D}}

\newcommand{\nn}{\nonumber}

\newcommand{\zglorentz}{{ \mathring{\glorentz}}}
\newcommand{\glorentz}{{ {\mathbf g}}}
\newcommand{\nablariem}{{ {D}}}

\newcommand{\griem}{{ {\mathfrak g}}}

\newcommand{\mcV}{{\mycal V}}

\DeclareFontFamily{OT1}{rsfs}{}
\DeclareFontShape{OT1}{rsfs}{m}{n}{ <-7> rsfs5 <7-10> rsfs7 <10->
rsfs10}{} \DeclareMathAlphabet{\mathscr}{OT1}{rsfs}{m}{n}

\newcommand{\mcM}{\mathscr M}

\newcommand{\GmcV}{ {G_\mcV}}
\newcommand{\GW}{ {G_W}}

\newcommand{\eq}[1]{\eqref{#1}}
\newcommand{\Eq}[1]{Equation~\eqref{#1}}

\newcommand{\bel}[1]{\begin{equation}\label{#1}}
\newcommand{\beal}[1]{\begin{eqnarray}\label{#1}}
\newcommand{\beaa}{\begin{eqnarray}}
\newcommand{\bean}{\begin{eqnarray}\nn}
\newcommand{\eeal}[1]{\label{#1}\end{eqnarray}}
\newcommand{\eea}{\end{eqnarray}}
\newcommand{\eeaa}{\end{eqnarray*}}

\newcommand{\be}{\begin{equation}}
\newcommand{\ee}{\end{equation}}

\DeclareFontFamily{OT1}{rsfs}{}
\DeclareFontShape{OT1}{rsfs}{m}{n}{ <-7> rsfs5 <7-10> rsfs7 <10->
rsfs10}{} \DeclareMathAlphabet{\mycal}{OT1}{rsfs}{m}{n}

\newcounter{mnotecount}[section]

\newcommand{\rmnote}[1]{}

\newcommand{\Ric}{\operatorname{Ric}}

\renewcommand{\sharp}{\#}

\newcommand{\Hess}{\operatorname{Hess}}

\newcommand{\Tr}{\operatorname{Tr}}

\newcommand{\R}{\mathbb R}

\newcommand{\mbbS}{\mathbb S}

\renewcommand{\to}{\rightarrow}
\renewcommand{\div}{\operatorname{div}}

\renewcommand{\hat}{\widehat}

\def\crn#1#2{{\vcenter{\vbox{
        \hbox{\kern#2pt \vrule width.#2pt height#1pt
           }
          \hrule height.#2pt}}}}

\newcommand{\newF}{\lambda}
\renewcommand{\H}{\mathbb H}

\renewcommand{\hbar}{{\overline h}}

\newcommand{\pre}[2]{{{\vphantom{#2}}^{#1}}\kern-.2ex{#2}}

\theoremstyle{plain}
\newtheorem{theorem}{\sc Theorem}[section]

\newtheorem{Proposition}[theorem]{\sc Proposition}

\theoremstyle{definition}

\newtheorem{remark}[theorem]{\sc  Remark\rm}

\numberwithin{equation}{section}

\date{\today}

\begin{document}

\title[Stationary spacetimes with negative $\Lambda$] {Non-singular spacetimes with a negative cosmological constant: III.  Stationary solutions  with matter fields}

\thanks{Preprint UWThPh-2016-30}

\author[P.T. Chru\'sciel]{Piotr T.~Chru\'sciel}

\address{Piotr
T.~Chru\'sciel, Faculty of Physics, University of Vienna, Boltzmanngasse 5, A1090 Wien, Austria}
\email{piotr.chrusciel@univie.ac.at} \urladdr{http://homepage.univie.ac.at/piotr.chrusciel/}

\author[E.  Delay]{Erwann
Delay} \address{Erwann Delay, Avignon Universit\'e, Laboratoire de Math\'ematiques d’Avignon (EA 2151)
F-84916 Avignon}
\email{Erwann.Delay@univ-avignon.fr}
\urladdr{http://www.math.univ-avignon.fr}

\author[P. Klinger]{Paul Klinger}

\address{Paul Klinger,  Faculty of Physics, University of Vienna, Boltzmanngasse 5, A1090 Wien, Austria}
\email{paul.klinger@univie.ac.at}

\begin{abstract}
We construct infinite-dimensional families of non-singular stationary
space times, solutions of {Yang-Mills-Higgs}-Einstein-Maxwell-Chern-Simons-dilaton-scalar field equations with a negative cosmological constant. The families include  an infinite-dimensional family of solutions with the usual AdS conformal structure at conformal infinity.
\end{abstract}

\maketitle

\tableofcontents

\section{Introduction}
 \label{section:intro}

There is currently considerable interest in the literature in space-times with a negative cosmological constant. This is fueled on one hand by studies of the AdS-CFT conjecture and of the implications thereof. On the other hand, these solutions are interesting because of a rich dynamical morphology: existence of periodic or quasi-periodic solutions, and of instabilities. All this leads naturally to the question of existence of stationary solutions of the Einstein equations with $
\Lambda<0$, with or without sources, and of properties thereof. Several families of such solutions have been recently constructed numerically~\cite{BjorakerHosotani,DiasHorowitzSantos,AstefaneseiRadu,BBSLMR,ToriiMaedaNarita,HerdeiroRadu,MaliborskiRostworowski}; cf.\ also~\cite{AFT} for a rigorous construction of static solutions of the Einstein-Vlasov equations with $\Lambda<0$.

In an accompanying paper~\cite{ChDelayEM} two of us (PTC and ED) have constructed an infinite dimensional family of  non-singular static
space times, solutions of the Einstein-Maxwell equations with a negative cosmological constant. These families include  an infinite-dimensional family of solutions with the usual AdS conformal structure at conformal infinity. The object of this work is to generalise the construction there to obtain a similar large family of singularity-free \emph{stationary} solutions of the \emph{Yang-Mills-Higgs-Einstein-Maxwell-dilaton-Chern-Simons-scalar field}
 equations, including a class of boson-star solutions with stationary metric but periodic complex scalar field.  We also show that our methods can be used to obtain solutions of a class of $f(R)$-theories.

Note that existence of such solutions of the Einstein-Maxwell-Yang-Mills-dilaton field equations with the  Kaluza-Klein value of the coupling constant  is a special case of the results in~\cite{ACD2}.

More precisely, we construct \emph{strictly} stationary solutions of the Einstein-matter field equations with a negative cosmological constant and with a smooth conformal boundary for a large class of matter models. Here we say that a space-time $(\mcM,\glorentz)$  is
strictly stationary if there exists on $\mcM$ a Killing vector field which is timelike everywhere. Such a solution is defined to be \emph{non-degenerate} if a certain operator associated with the linearisation  of the field equations is an isomorphism, cf.\ Section~\ref{s2XII16.1} for a precise definition. An example of a non-degenerate solution is anti-de Sitter space-time. Our solutions are constructed using an implicit function theorem near a non-degenerate vacuum metric  $(\mcM,\zglorentz)$.
We also construct solutions with a time-periodic complex scalar field accompanied by time-independent metric and Maxwell fields.
The solutions are uniquely determined by certain freely prescribable coefficients in the asymptotic expansion of the metric, of the Yang-Mills or Maxwell fields, of the dilaton field and of the scalar fields. Here uniqueness is guaranteed in a neighborhood of the metric $\zglorentz$. In this way we obtain infinite dimensional families of solutions with, if desired, the same conformal structure at infinity as the initial static vacuum metric $\zglorentz$.

By switching-off some free data at the conformal boundary, or setting to zero one of the coupling constants, one can obtain non-trivial solutions of the Einstein-Yang-Mills equations, or Einstein-scalar field equations, or static Einstein-Maxwell-Chern-Simons-dilaton solutions, etc.
In particular we establish rigorously
existence of Einstein-Yang-Mills solutions in  near-AdS configurations, as constructed numerically in~\cite{BjorakerHosotani,BML}, and in fact we provide a much larger family of such solutions.

The method is a conceptually-straightforward repetition of the arguments in~\cite{ChDelayStationary,ChDelayEM}, so that our presentation will be suitably sketchy: we will only provide details at places which require technical or calculational changes.

Our hypothesis of \emph{strict} stationarity excludes black hole solutions. The extension of our analysis to black holes will be discussed elsewhere~\cite{CDKprep}.

A similar construction works near any non-degenerate stationary solutions of the equations under consideration, provided that the linearisations of the matter equations lead to isomorphisms. This last property appears to require a case-by-case analysis of the solutions at hand.

\section{Stationary Einstein-Maxwell-Chern-Simons-dilaton-scalar field equations in $n+1$ dimensions}
 \label{s18XI16.1}

We consider the Einstein equations for a metric
$$
 \glorentz
  = g_{\mu\nu}dx^\mu dx^\nu
$$
in space-time dimension $n+1$, $n\ge 3$,
\bel{11X16.1}
\Ric(\glorentz ) -\frac{{}\Tr_{ \glorentz }  \Ric(\glorentz )}2 {} \glorentz +\Lambda \glorentz  = 8 \pi G T
 \,,
\ee
where $T$ is the energy-momentum tensor of matter fields. A constant rescaling of $\glorentz $ allows one  to normalise a negative
cosmological constant to
\bel{16XI16.11}
\Lambda=-\frac{n(n-1)}2
 \,,
\ee
and we will often use this normalisation. The space-time manifold $\mcM$ will be taken of the form $\R\times M$, with the $\R$ coordinate running along the orbits of a Killing vector field which is \emph{timelike everywhere}.

In the Einstein-Maxwell-Chern-Simons-dilaton-scalar field case we have~\cite{Kim:2016dik,BKNR}
\beal{21XI16.1+}
 \lefteqn{
    T_{\alpha\beta}=\frac{1}{8\pi G}\big[\frac{1}{2}\partial_\alpha \phi \partial_\beta \phi+2 W(\phi)F_{\alpha\mu}F_\beta{}^\mu
 }
&&
\\
 &&
 \phantom{xxxx}
 - \glorentz _{\alpha\beta}\big(\frac{1}{4}(\nabla \phi)^2+\frac{W(\phi)}{2}|F|^2+\frac{1}{2}{\mycal V}(\phi)\big)\big]\,,
 \nn
\eea
with action
 \bel{21XI16.2}
 S=\int d^{n+1}x \frac{\sqrt{-\det{}\glorentz }}{16\pi G}\left[R(\glorentz)-2\Lambda-W(\phi)|F|^2-\frac{1}{2}(\nabla\phi)^2-{\mycal V}(\phi)\right] +S_{\mathrm{CS}}\,,
\ee
where $R(\glorentz)$ is the Ricci scalar of the metric $\glorentz$ and {where, in odd space-time dimensions, $S_{\mathrm{CS}}$ is the Abelian Chern-Simons action:
\bel{8XII16.2}
 S_{\mathrm{CS}} = \left\{
                     \begin{array}{ll}
                       0, & \hbox{$n$ is odd;} \\
 \displaystyle
                       \frac{\lambda }{16\pi G}\int A\wedge \underbrace{F\wedge \cdots \wedge F}_{k \text{ times}}, & \hbox{$n=2k$,}
                     \end{array}
                   \right.
\ee
for a constant $\lambda \in \R$.
}
We will assume that $W$ and ${\mycal V}$ are smooth functions, and require
\bel{20XI16.16}
 W(0) = 1
 \,,
 \quad
 {\mycal V}(0) = 0
\ee
(note that this differs from the conventions of the accompanying paper~\cite{ChDelayEM}, where $\phi \equiv 0$ and where the normalisation $W\equiv 1/2$ has been used).

We can view $\phi$ as taking values in a Euclidean $\R^{N+1}$ for some $N\ge 0$, with the first component $\phi^1$ corresponding to the dilaton field, and with $W$ depending only upon $\phi^1$. Then the remaining components $(\phi^2, \ldots,\phi^{N+1})$ of $\phi$  describe $N$ minimally-coupled scalar fields, possibly interacting with each other through the potential $\mcV$ which might or might not depend upon $\phi^1$.

Taking $N=0$, $\mcV\equiv 0$, $\phi =  2 u$, $W(u)= e^{-2a u}$ for a constant $a
 \in \R$, and setting the Chern-Simons coupling constant $\lambda$ to zero  one obtains the usual Einstein-Maxwell-dilaton equations with action~\cite{GHT}
\bel{21XI16.2+}
 S=\frac{1}{16\pi G}\int d^{n+1}x \sqrt{-\det{}\glorentz }\left[R-2\Lambda - e^{-2a u}|F|^2-{2}(\nabla u)^2\right]\,.
\ee

Similarly, we can view $F$ as taking values in a Euclidean $\R^{N_1}$ for some $N_1\ge 0$, in which case we obtain a collection of Abelian Yang-Mills fields
$F^B_{\mu\nu}dx^\mu \wedge dx^\nu $, $B=1,\ldots, N_1$.
The Chern-Simons action \eq{8XII16.2} can then be replaced by
\bel{5I17.1}
 S_{\mathrm{CS}} = \left\{
                     \begin{array}{ll}
                       0, & \hbox{$n$ is odd,} \\
 \displaystyle
                       \frac{1 }{16\pi G}\lambda_{B B_1 \ldots B_k}\int A^B\wedge F^{B_1}\wedge \cdots \wedge F^{B_k} , & \hbox{$n=2k$,}
                     \end{array}
                   \right.
\ee
for a set of constants $\lambda_{B B_1 \ldots B_k}$, totally symmetric in the last $k$ indices.

Our  analysis extends to general Yang-Mills-Higgs-dilaton-Chern-Simons fields in the obvious way, by replacing $\partial \phi$ by a gauge-covariant derivative. This is addressed in Section~\ref{ss2I17.3} below.

\section{Definitions, notations and conventions}
 \label{s2XII16.1}

Our definitions and conventions are identical to those in~\cite[Section~2]{ChDelayStationary}. In particular $\rho$ denotes a non-negative smooth function which has nowhere vanishing gradient near $\partial M$ and vanishes precisely at $\partial M$.

{Recall that the linearisation of the Ricci tensor in dimension $(n+1)$ equals~\cite[Equations~(1.180a)-(1.180b), p.~64]{besse:einstein}
\bel{31XII16.11}
 \frac 12 \big( \Delta_L h - 2 \delta^*\delta h - D d (\mathrm{tr} h) \big)
 \,,
\ee
or in index notation
$$
 \frac 12 \big( \Delta_L h_{ij}+D_i D^k h_{kj}+D_jD^k h_{ki}-D_iD_j h^k{}_k)
 \,,
$$
where $\Delta_L h $ is
the  Lichnerowicz Laplacian  acting on the symmetric
two-tensor field $h$,  defined
as~\cite[\S~1.143]{besse:einstein}
$$
\Delta_Lh_{ij}=-D^kD_kh_{ij}+R_{ik}h^k{_j}+R_{jk}h^k{_i}-2R_{ikjl}h^{kl}\,.
$$

An explicit form of $\Delta_L  $ for a metric of the form \eq{gme1} below can be read-off from the formulae in~\cite[Appendix~A]{ChDelayStationary}.
}

We will say that a metric $\griem$ is {\it non-degenerate} if $\Delta_L+2n$ has no
$L^2$-kernel. Large classes of non-degenerate Einstein
metrics are described in~\cite{Lee:fredholm,ACD2,mand1,mand2}.

\input{method}

\section{Static {metrics}}
 \label{ss18XI16.1}
 \label{s11XII16.1}

In this section we present the construction of static solutions of the equations at hand. Strictly speaking, the results in this section are a special case of those in Section~\ref{ss18XI16.2+}, but it appears instructive to present them separately, taking into account that the analysis here is computationally less demanding than the general case.

Assuming staticity, in adapted coordinates the metric $\glorentz $ becomes
\bel{18XI16.1}
 \glorentz  = -V^2 dt^2 + g
 \,,
  \quad
  \partial_t V = 0 = \partial_t g
 \,.
\ee
Equations~\eq{11X16.1} lead to the following set of equations,
where we denote by $\nablariem$ the covariant derivative of $g$ (recall that $\nabla$ is the covariant derivative of $\glorentz$), and where $R_{ij}$ is the Ricci tensor of $g$:
\beal{7XI16.2}
 &
 \label{15XI16.21}
  \displaystyle
 R_{ij} = V^{-1} \nablariem _i\nablariem _j V + \frac{2\Lambda}{n-1} g_{ij} + 8\pi G\left( T_{ij} - \frac {\Tr_{{}\glorentz } T } {n-1} g_{ij}\right)
 \,,
 &
\\
  &
  \displaystyle
 \phantom{xxx}
  V \nablariem ^i \nablariem _i V = 8\pi G V^2 \left( T_{\alpha \beta} N^\alpha N^\beta + \frac{\Tr_{{}\glorentz } T } {n-1}
   \right)
   - V^2 \frac{2 \Lambda}{n-1}
  \,,
  \quad
   T_{0i}=0
  \,,
  &
\eeal{14XI16.1}
where $N^\alpha\partial_\alpha$ is the $\glorentz$-unit timelike normal to the level sets of $t$.
Choosing $\Lambda$ as in \eq{16XI16.11}, taking into account the Maxwell equations and the scalar field equations, together with
\bel{4XII16.1}
 \partial_t F_{\mu\nu} =0 =  \partial_t \phi
 \,,
\ee
one is led to the system
%
%
\begin{equation}
 \label{22XI16.1}
\left\{\begin{array}{l} V( { -\Delta_g}  V + n V) = - 2W(\phi)F_{0i}F_0{}^i + \frac{V^2}{n-1}({\mycal V}(\phi) - W(\phi) |F|^2)
\,,\\
    R_{ij} + ng - V^{-1}\nablariem _i\nablariem _jV = \frac{1}{2}\partial_i\phi \partial_j\phi+2F_{i\alpha}F_j{}^\alpha W(\phi) + \frac{g_{ij}}{n-1}({\mycal V}(\phi) - W|F|^2)\,,
\\
\frac{1}{V\sqrt{\det g}}\partial_\mu (V \sqrt{\det g} W(\phi) F^{\mu\nu})+ B_{\mathrm{CS}}^\nu=0\,,
\\
\frac{1}{V\sqrt{\det g}}\partial_i (V \sqrt{\det g}g^{ij}\partial_j \phi) - W'(\phi) |F|^2 - {\mycal V}'(\phi)=0
 \,,
\\
W(\phi) F_{0j}F_i{}^j=0
 \,,
\end{array}\right.
\end{equation}
where  $W'$ and $\mcV'$ are understood as differentials of $W$ and $\mcV$ when $\phi$ is $\R^{N+1}$ valued with $N\ge1$, and where the Chern-Simons source-term $B_{\mathrm{CS}}^\nu$ is given by
%
%
\bel{8XII16.3}
 \displaystyle
B_{\mathrm{CS}}^\nu = \left\{
                     \begin{array}{ll}
                       0, & \hbox{$n$ is odd;}
 \\
 \displaystyle
                  -  \frac{ \lambda }{2^{k+2}}
\epsilon^{\nu\alpha_1\beta_1\cdots\alpha_k\beta_k} F_{\alpha_1\beta_1} \cdots F_{\alpha_k\beta_k}, & \hbox{$n=2k$.}
                     \end{array}
                   \right.
\ee

\subsection{Purely electric configurations}
 \label{ss2I16.1}
One way of satisfying the last equation in \eq{22XI16.1} is to assume a purely electric Maxwell field:
\bel{4XI116.5}
 F=d(Udt)
 \,,
 \quad
  \partial_t U =0
  \,.
\ee
(Purely magnetic configurations will be considered in Section~\ref{ss2I17.2} below, while configurations with both electric and magnetic fields can be obtained by applying duality rotations to the Maxwell field at the end of the construction.)
\Eq{4XI116.5} leads to the following form of \eq{22XI16.1}:
\begin{equation}
 \label{18XI16.3}
\left\{\begin{array}{l}
 V( { -\Delta_g}  V+n V)=-\frac {2(n-2)}{n-1}  W(\phi)|dU|_g^2   + \frac{1}{n-1} V^2{\mycal V}(\phi)
 \,,
\\
    \Ric(g)+n g-V^{-1}\Hess_gV= \frac 12 d\phi \otimes d\phi
      + \frac{1}{n-1} {\mycal V}(\phi)g
     \\
     \phantom{xxxxxxxxxxxxxxx}
      +2 W(\phi)V^{-2}\left(-dU\otimes dU+\frac{1}{n-1}|dU|_g^2g\right)\,,
\\
\div_g (V^{-1} W(\phi) \nablariem  U)=0\,,
\\
V^{-1}\div_g  (V D \phi )  + 2 V^{-2} |dU|^2_g W'(\phi) - {\mycal V}'(\phi)=0
\end{array}\right.
\end{equation}
(the Chern-Simons term drops out because the purely spatial components of $F$ vanish).

When $\phi\equiv 0$ the $U$-equation coincides with that in~\cite{ChDelayEM}, thus Theorem~3.3 there with $ s=-1$ applies. By continuity it still applies for all fields $\phi$ which are sufficiently small in $ C^{k+1,\alpha}_\epsilon $, with any $\epsilon>0$.
This motivates us to seek again solutions with $U$ of the form
\bel{4XI16.1}
 U = \hat U + O(\rho)
 \,,
\ee
where $\hat U$ is smooth-up-to-boundary on $\overline M$. (Here two comments are in order: First, the key information is contained in the function $\hat U|_{\partial M}$ defined on the boundary, but it is useful to invoke a function $\hat U$ defined on $\overline M$, which avoids the issue of considering extensions to $M$ of functions defined on $\partial M$. Next, the uniqueness part of our analysis below implies that solutions with
$\hat U |_{\partial M}= c$ for some constant $c\in \R$ lead  to configurations with $U\equiv c$, hence trivial Maxwell fields. {In other words, in \eq{4XI16.1} only $\hat U$ modulo constants matters as far as physically relevant fields are concerned. Nevertheless,  different $c$'s lead indeed to different fields $U$.)

We will seek a solution such that
\bel{11I17.1}
 \mbox{ $\phi\to 0$ as $\rho\to 0$.}
\ee
Note that with the choice \eq{4XI16.1} the  coefficient  $  2 V^{-2} |dU|^2_g$ appearing in the $\phi$-equation will be $O(\rho^4)$.
Assuming \eq{11I17.1}} we must have $\mcV'(0)=0$, so that
in the scalar case
the indicial exponents for the $\phi$-equation (cf., e.g.,~\cite{Lee:fredholm})
will be solutions of the equation
\bel{1XII16.4}
 \sigma(\sigma-n)- \mcV''(0) = 0
 \quad
 \Longleftrightarrow
 \quad
  \sigma_\pm  = \frac{ n}2 \pm \sqrt{\frac{n^2  }{4} + \mcV''(0)}
  \,.
\ee
When $\phi$ is $\R^{N+1}$ valued, with $N\ge 1$, we obtain a collection of indicial exponents, with $\mcV''(0)$ in \eq{1XII16.4} replaced by the eigenvalues of $\mcV''(0)$.
Now, to ensure useful properties of the operator associated with the equation for $\phi$ we need all $\sigma_\pm$ to be  real with $\sigma_+\ne \sigma_-$, which leads to the condition
\bel{5XII16.2}
    -\frac {n^2}4 < \mcV''(0)
 \,,
\ee
understood as a matrix inequality for the Hessian of $\mcV$ at $\phi=0$ when $N\ge 1$. After diagonalising $\mcV''(0)$, each component
in the diagonalising basis of the solutions will have the asymptotic behaviour
\bel{5XII16.1}
 \phi = \hat \phi \rho^{\sigma_-} + o(\rho^{\sigma_-})
 \,.
\ee
possibly with $ \hat \phi |_{\partial M} \equiv 0$, and note that we have
\bel{21XII16.22}
\phi\to_{\rho\to 0} 0\ \mbox{and} \   \hat \phi |_{\partial M} \not\equiv 0
 \quad
    \Longrightarrow
 \quad
 \big(
 \sigma_->0
    \
 \Leftrightarrow
    \
 \mcV''(0)<0
    \big)
  \,.
\ee

The properties of solutions of the $\phi$-equation depend now upon whether or not $W'(0)=0$.

Let us first  assume that $W'(0)=0$  and that $\mcV''(0)$
is not an $L^2$-eigenvalue of the operator $\Delta_{\mathring\griem}$ (see Remark~\ref{R24III17.1} below).

By Theorem~\ref{Tisofunction}, Appendix~\ref{sA16III17.1} with $s=1$ the operator
$$
 \phi \mapsto \mathring V^{-1}\div_{\mathring g}(\mathring V D\phi)-\mcV''(0)\phi
$$
that arises by linearising the $\phi$ equation is an isomorphism from
$C^{k+2,\alpha}_{\delta}$ to $C^{k,\alpha}_{\delta}$ for $\delta\in(\sigma_-,\sigma_+)$.}
We will then have a non-trivial solution $\phi\not \equiv 0$ tending to zero when $\rho\to 0$ if and only if
$\hat \phi |_{\partial M} \not \equiv 0
$.
Since the case $\mcV''(0)\ge0$  leads then to solutions which do not tend to zero at the conformal boundary, the hypothesis that $\mcV''\ge 0$ and $W'(0)=0$ leaves us with Maxwell matter fields only.

Assuming that $\mcV''(0)< 0$ and $W'(0)=0$, it remains to check that the  source terms in the remaining equations are compatible with the isomorphism ranges of the relevant operators.
For this  it is convenient to rewrite the $V$-equation as
\bel{21XII16.21}
  { -\Delta_g}  V+(n + \frac{1}{n-1} {\mycal V}(\phi))  V=-\frac {2(n-2)}{n-1} V^{-1} W(\phi)|dU|_g^2
 \,.
\ee
Since the coefficient $\frac{1}{n-1} {\mycal V}$ goes to zero at the boundary when $\phi$ does, we obtain the same indicial exponents as when $\phi \equiv 0$, and thus again an isomorphism for $\phi$ small enough. No new conditions arise from the remaining equations either.

Summarising, for $\hat \phi|_ {\partial M} \ne 0$ in view of \eq{5XII16.2} we must have
\bel{5XII16.2+}
    - {n^2} < 4\mcV''(0) < 0
 \,.
\ee
If $\phi$ has more than one component then the above inequalities apply with $\mcV''(0)$ replaced by the relevant eigenvalue of the Hessian of $\mcV(0)$ (in this case it is convenient to work in a diagonalising basis).

If $W'(0)\ne 0$ the situation is different, as then the scalar fields $\phi$ are driven both by the term $ 2 V^{-2} |dU|^2_g W'(0)$ and by $\hat \phi|_{\partial M} $. First, if $\hat U|_{\partial M} \equiv 0$, then $U\equiv 0$, and if $\hat \phi|_{\partial M}$ vanishes as well, then $\phi \equiv 0$, and we are in vacuum. On the other hand,  if $\hat \phi|_{\partial M} \equiv 0$ and $\hat U|_{\partial M} \not \equiv 0$ then we have non-trivial solutions with the following asymptotic behaviour:
\bel{18XII16.11}
 \phi =\left\{
         \begin{array}{ll}
           O(\rho^{\sigma_+}) & \hbox{if $\sigma_+< 4$;} \\
           O( \rho^{4}\ln \rho) & \hbox{if $\sigma_+=4$;} \\
           O(\rho^{4}) & \hbox{if $\sigma_+> 4$ .}
         \end{array}
       \right.
\ee
An analysis similar to that of the case $W'(0)=0$, taking into account that $\sigma_+>n/2$ under \eq{5XII16.2}, shows that non-trivial $\phi$'s tending to zero at the boundary will be obtained when $\hat U|_{\partial M} \not \equiv 0$ if, in the one-component scalar field case,
\bel{5XII16.3}
 \left\{
    \begin{array}{ll}
          -{n^2} < 4\mcV''(0) < 0, \ \mbox{or } & \\
           \mcV''(0)  \ge 0  \ \mbox{and $\hat \phi|_{\partial M} \equiv 0$.} & \
    \end{array}
  \right.
\ee

 The question of  stability of the solutions with $\mcV''(0)<0$ is unclear (compare~\cite{ToriiMaedaNarita}, and there are in fact hints that some solutions with $\mcV''(0)<0$ might be stable~\cite{BHJM}), but this is irrelevant from the point of view of the question existence of static or stationary solutions, which is our main interest in this work.

{One can now proceed exactly as in~\cite{ChDelayEM}
to obtain the following: Consider the field equations associated with the action \eqref{21XI16.2} for time-independent fields, with
\bel{8XII16.1}
 W(0)=1
 \,,
  \quad
  \mcV(0)=0
   =
  \mcV'(0)
  \,,
   \quad
   \mcV''(0)> -n^2/4
  \,.
\ee
Let
$$
 \mathring \glorentz = -\mathring V^2 dt^2 + \mathring g
$$
be smoothly conformally compactifiable at $\R\times \partial M$ and satisfy the vacuum Einstein equations with a negative cosmological constant.
Then:

\begin{Proposition}
   \label{P21XII16.1}
Suppose that  $(S^1\times M, \mathring V^2 dt^2 +\mathring g)$ is nondegenerate, and  that $\mcV''(0)$
is not an $L^2$-eigenvalue of $\Delta_{\mathring \griem}$ (cf.\ Remark~\ref{R24III17.1} below). Under \eq{8XII16.1}, assume that
\begin{enumerate}
  \item
$
\mcV''(0) < 0
$
with  $\hat U$ and  $\hat \phi$ which are smooth functions on $\overline M$ sufficiently close to zero, or

  \item
$\hat \phi \equiv 0$ and  $\hat U$ is a smooth function  on $\overline M$ sufficiently close to zero.
\end{enumerate}

Then there exists a static solution of the equations with $U$ as in \eq{4XI16.1} and $\phi$ as in \eq{5XII16.1}.
The solutions are uniquely determined, within the class of static solutions belonging to some neighbourhood of $\mathring \glorentz$, by $\hat U|_{\partial M} $
   and $\hat \phi|_{\partial M} $, with all fields having a polyhomogeneous expansion at $\partial M$.
\qed
\end{Proposition}

If $\hat \phi |_{\partial M} \equiv 0 $ and $W'(0)=0$ then $\phi\equiv 0$, so that we obtain the solutions of Einstein-Maxwell equations already constructed in~\cite{ChDelayEM}.

The reader will find some more information about the asymptotics of the fields in Section~\ref{s10I17.1}.

We emphasise that uniqueness is in the gauges implicitly defined above; for instance,  two solutions with $U$-fields differing by a constant are considered distinct, even though they define of course the same Maxwell field $F_{\mu\nu}$. Furthermore, uniqueness is  up-to diffeomorphisms which are {the identity at the boundary}  in any case. {It is conceivable that uniqueness holds for arbitrary diffeomorphisms, but this does not follow from our arguments.}

\subsection{Purely magnetic configurations}
 \label{ss2I17.2}
Another way to satisfy the last equation in \eqref{22XI16.1} is to consider purely magnetic fields, i.e.
\begin{equation}
    F=d(A_i dx^i)\,,\qquad \partial_t A_i=0\,.
\end{equation}
This implies $
B_{\mathrm{CS}}^i=0$
and leads to the following  matter equations,
\begin{equation}\label{matterequations_static_mag}
\left\{\begin{array}{l}
 \nablariem _j (V W g^{jk}g^{il}(A_{l,k}-A_{k,l}))=0\,,\\
 V^{-1}\nablariem _i (V g^{ij} \partial_j \phi) - 2 W'(\phi) (A_{j,i}-A_{i,j})g^{ik}g^{jl}A_{l,k} - {\mycal V}'(\phi)=0\,,\\
 B_\mathrm{CS}^0=0 \,,
\end{array}\right.
\end{equation}
with the Einstein equations taking now the form
\begin{equation}\label{mainequations_static_mag}
\left\{\begin{array}{l}
V(-\Delta_g V + n V) = \frac{V^2}{n-1}({\mycal V} - 2W  W'(\phi) (A_{j,i}-A_{i,j})g^{ik}g^{jl}A_{l,k})\, ,\\
\begin{split}
    R_{ij}+n g_{ij}-V^{-1}\nablariem _i\nablariem _j V=&
     \frac{1}{2}(\partial_i \phi)(\partial_j \phi)+2W (A_{k,i}-A_{i,k})(A_{l,j}-A_{j,l})g^{kl}
\\
    &+\frac{g_{ij}}{n-1}({\mycal V}-2W W'(\phi) (A_{j,i}-A_{i,j})g^{ik}g^{jl}A_{l,k})\, .
\end{split}
\end{array}\right.
\end{equation}
To satisfy the last equation of \eqref{matterequations_static_mag} one might as well assume that the Chern-Simons coupling constant $\lambda $ vanishes.
The first line of \eqref{matterequations_static_mag} can be rewritten, after introducing $\sigma_F$ as in \eq{10XII16.4}, in the form
\beal{21XII16.2}
 \lefteqn{\phantom{xxx} 0 =
\nablariem _j (V W g^{jk}g^{il}(\nablariem _kA_{l }))}
 &&
\\
 \nn
 && \phantom{xx}-  \nablariem ^k (V W  ) \nablariem ^i A_{k}
- VW R^i{}_\ell A^\ell -   V W   \nablariem ^i \left(\sigma_F -V^{-1} A^k \nablariem_ kV \right)
 \,.
\eea
If we  develop (\ref{21XII16.2}) and drop $\sigma_F$, the operator acting on $A$
becomes
$$
 VW B(A)_i+VD^kW(D_kA_i-D_iA_k)+WA^kD_iD_kV-VW R_{ik}A^k
\,,
$$
where $B$ is the operator of Lemma A.3 in~\cite{ChDelayStationary}.
The operator
$$
\mathcal P:=B+(V^{-1}DDV-\Ric(g))
$$
appears as part of the $(n+1)$-dimensional Riemannian Hodge Laplacian $D_{\griem}^*D_{\griem}+\Ric(\griem)$~\cite{ChDelayStationary}.
Recall that, in coordinates,  the characteristic indices for
$D_{\griem}^*D_{\griem}+\Ric(\griem)\sim D_{\griem}^*D_{\griem}-n$ belong to $\{-1, n-1\}$ in the normal direction and to $\{0, n-2\}$ in the tangential one (cf., e.g.~\cite[Section 2.3]{AubryGuillarmou}).
We deduce from Theorem C(c) and Corollary 7.4 of~\cite{Lee:fredholm} that if there are no harmonic forms in  $L^2$   for   {($S^1\times M,\mathring\griem)$},
then $\mathcal P$ 
  will be an isomorphism from $C^{k+2,\alpha}_\delta$ to
$C^{k,\alpha}_\delta$ for
\bel{indicesisoA}
\left|\delta-\frac n2\right|<\frac n 2-1
 \,.
\ee

Let us denote by
$$(\rho,x^a)
$$
 local coordinates near $\partial M$. From what has been said
one expects solutions to take the form
\bel{21XII16.1}
 A = \big( \hat A_\rho \rho^{-1} + O(1)\big) d\rho
  +  \big( \hat A_a  + O(\rho)\big) dx^a
 \,,
\ee
with smooth-up-to-boundary functions $\hat A_i$. As discussed in detail at the end of  Section~\ref{ss18XI16.2}, we will have $ \sigma_F\equiv 0$ if and only if
 $\nabla^\mu A_\mu$ tends to zero as $\rho$ tends to zero. Now
\bel{20XII16.1}
 \nabla^\mu A_\mu=\frac{\partial_i(g^{ij}\sqrt{\det{g}}V A_j)}{\sqrt{\det g}V}=  \partial_\rho(\rho^{1-n} A_\rho)\rho^{1+n}
 +O(\rho)
 \,.
\ee
This will be satisfied  if and only if $\hat A_\rho |_{\partial M} \equiv 0$, without any restrictions on $\hat A_a|_{\partial M} dx^a $. We conclude that:

\begin{Proposition}
   \label{P21XII16.2}
{Under the hypotheses of Proposition~\ref{P21XII16.1}, suppose moreover that the Einstein metric $(S^1\times M, \mathring V^2 dt^2 + \mathring g)$  has no harmonic one-forms which are in $L^2$.  Then the conclusions of}
Proposition~\ref{P21XII16.1} hold with $U$ replaced by $A_i dx^i$ of the form \eq{21XII16.1} with $\hat A_\rho |_{\partial M} \equiv 0$, so that  $\hat U|_{\partial M} $ is replaced by $\hat A_a|_{\partial M}  dx^a$.
\qed
\end{Proposition}

As discussed in Appendix~\ref{s29XII16.2}, it follows from~\cite{CarronPedon} that there are no $L^2$ harmonic one-forms on $S^1\times M$
equipped with the Riemannian  counterpart of the AdS metric $\mathring\griem$, so the same is true
for nearby metrics.

Further remarks concerning asymptotics and total energy are to be found in Section~\ref{s10I17.1} below.

\subsection{Yang-Mills-Higgs fields}
 \label{ss2I17.3}
The analysis so far readily generalises to Yang-Mills-Higgs-Chern-Simons fields.  Here one often assumes that the Lie algebra $\mathfrak G$ of the structure group $G$ admits a positive-definite scalar product, but this is not needed in our considerations.

We denote by $A=A_\mu dx^\mu$ the Yang-Mills connection, with $A_\mu$ taking values in $\mathfrak G$, and by $F=\frac 12 F_{\mu\nu}dx^\mu \wedge dx^\nu$ its curvature:
$$
  F_{\mu\nu}  = \partial_\mu A_\nu - \partial_\nu A_\mu +  [A_\mu,A_\nu]
  \,.
$$
The scalar fields $\phi$ are allowed to be coupled to $A$ in the usual way, with derivatives involving $\phi$ replaced by gauge-covariant derivatives
$$
 \partial_\mu \phi \mapsto \mcD_\mu \phi:=\partial_\mu \phi + T(A_\mu) \phi
 \,,
$$
where $T$ is the linear map determined by the relevant representation; e.g., if  $\phi$ are sections of a bundle associated to the adjoint representation, then $T(A_\mu)\phi= [A_\mu,\phi]$. 

We suppose that the $\mathfrak G$-valued current vector $j^\nu$ appearing in the Yang-Mills equations,
\begin{eqnarray}
 \label{2I17.5}
&
 \mcD_\mu F^{\mu\nu}:= \nabla_\mu F^{\mu\nu}+ [A_\mu, F^{\mu\nu}] = j^\nu
 \,,
&
\\
&
 \mbox{
depends only upon $A$, $\partial A$, $\phi$, $\partial \phi$, $g$, $\partial g$,}
&
\\
& \mbox{is at least quadratic in all the fields,}
&
 \label{3I17.1}
\eea
 and satisfies the obvious compatibility conditions arising from
\bel{2I17.5+}
 \mcD_\nu \mcD_\mu F^{\mu\nu} =  \mcD_{[\nu} \mcD_{\mu]} F^{\mu\nu} = [F_{\mu\nu},F^{\mu\nu}] = 0
 \quad
  \Longrightarrow
  \quad
   \mcD_\nu j^\nu = 0
  \,.
\ee
Equivalently
\bel{2I17.6}\nabla_\nu (j^\nu - [A_\mu, F^{\mu\nu}]) = 0
  \,.
\ee
More precisely, we will need that
\bel{2I16.7}
\mbox{  $\mcD_\nu j^\nu = 0$
whenever the field equations for $\phi$ are satisfied.}
\ee
This will be automatically satisfied by currents arising from gauge-invariant Lagrange functions, and by the current arising from Chern-Simons terms~\cite{ChernSimons}, whether Abelian or non-Abelian.

Let us, first, assume that the Yang-Mills principal bundle, say $P(G)$, is trivial, so that $F$ is globally defined as a two-form. The case where $A = U dt$, $\partial_t U =0$, where $U$ is $\mathfrak G$-valued, works exactly as in the Maxwell case, leading immediately to an obvious Yang-Mills equivalent of Proposition~\ref{P21XII16.1} whenever $P(G)= G\times M$.

A purely-magnetic Yang-Mills potential, $A_0\equiv 0$, $\partial_t A_i \equiv0$, does not require much more work.
Under our conditions, \eq{matterequations_static_mag}-\eq{mainequations_static_mag} are only modified by terms which are at least quadratic in the fields and which are lower order in terms of derivatives. Such terms, when small enough in relevant norms, do not affect the argument: One can view $F$ as a collection of several electric fields, introduce a vector-valued gauge-source function $\sigma_F$ (one such function for each component of $F$), and conclude as before. In other words, assuming \eq{2I17.5}-\eq{3I17.1} and  \eq{2I16.7}, we have established the Yang-Mills equivalent of Proposition~\ref{P21XII16.2} for trivial $G$-bundles.
This establishes, for small $A$, existence of the solutions constructed numerically in~\cite{BjorakerHosotani}, and in fact provides a much larger family of such solutions. (An even larger family of such solutions results from Theorem~\ref{T27XII16.1} below with $\hat \varphi|_{\partial M}=0=W'(0)$.)

Summarising, we have proved:

\begin{Proposition}
  \label{P6III17.1}
 The conclusions of Propositions~\ref{P21XII16.1} and \ref{P21XII16.2} hold when Maxwell and scalar fields are replaced by  Yang-Mills and Higgs  fields on a trivial gauge bundle.
\qed
\end{Proposition}

Non-trivial bundles can be handled by introducing a suitably regular  background Yang-Mills connection $\mathring A$. This leads to a globally defined $\mathfrak G$-valued one-form $A-\mathring A$, and a corresponding globally defined $\mathfrak G$-valued $\sigma_F$ function
$$
 \sigma_F:= \mathring \mcD^\mu (A_\mu - \mathring A_{\mu}):=  \nabla^\mu (A_\mu - \mathring A_{\mu})
  + [\mathring A^\mu,  A_\mu - \mathring A_{\mu} ]
 \,.
$$
The existence argument goes through  if one moreover assumes that
\begin{enumerate}
  \item there are no covariantly-constant Higgs fields $\phi$ which are in $L^2(S^1\times M)$, and that
      \item there are no $\mathfrak G$-valued harmonic forms which are in $L^2(S^1\times M)$.
\end{enumerate}

\subsection{Time-periodic scalar fields}
 \label{ss19XII16.1}

Let us allow complex-valued $\phi$'s, and assume that
\bel{12I17.1}
 \mbox{$\mcV (\phi)=\GmcV (|\phi|^2)$ and $W(\phi)=\GW (|\phi|^2)$}
\ee
for some differentiable functions $\GmcV$ and $\GW$, with the term $(\nabla \phi)^2 $ in the action replaced by $ \nabla^\alpha \overline{\phi}\nabla_\alpha \phi $, where $\overline{\phi}$ is the complex conjugate of $\phi$.
Considering, as in~\cite{Kaup}, a time dependent field of the form
\bel{21XII16.25}
 \phi(t,x)=e^{i\omega t}\psi (x)
 \,,
 \quad \ \mbox{with} \
 \omega\,,\, \psi(x) \in \R
 \,,
\ee
(compare~\cite{DiasHorowitzSantos})
leads to
\begin{equation}
 \label{19XII16.1}
\left\{\begin{array}{l} V( { -\Delta_g}  V + n V) = -\frac{1}{2}\omega^2 \psi ^2 - 2\GW F_{0i}F_0{}^i + \frac{V^2}{n-1}(\GmcV  - \GW   F_{\alpha\beta}F^{\alpha\beta})
\,,\\
\begin{split}
    R_{ij} + ng_{ij} - V^{-1}\nablariem _i\nablariem _jV
    =&  \frac{1}{2}\partial_i\psi  \partial_j\psi +2F_{i\alpha}F_j{}^\alpha \GW   \\
    &+ \frac{g_{ij}}{n-1}(\GmcV  - \GW  F_{\alpha\beta}F^{\alpha\beta})\,,
\end{split}
\\
\frac{1}{V\sqrt{\det g}}\partial_j (V \sqrt{\det g} \GW   F^{j\nu})+ B_{\mathrm{CS}}^\nu=0\,,
\\
\frac{1}{V\sqrt{\det g}}\partial_i (V \sqrt{\det g}g^{ij}\partial_j \psi) + V^{-2} \omega^2 \psi - \big( \GW'|F|^2
 + \GmcV'\big) \psi =0
 \,,
\\
\GW F_{0j}F_i{}^j=0
 \,.
\end{array}\right.
\end{equation}
We see that the indicial exponents for the system remain unchanged, so that the existence and uniqueness theory with $\omega=0$, presented above, applies without changes for all sufficiently small $\omega\in \R$:

\begin{Proposition}
  \label{P21XII16.3}
 The conclusions of Propositions~\ref{P21XII16.1} and \ref{P21XII16.2} hold for all sufficiently small $\omega \in \R$ where $\phi$ takes the form \eq{21XII16.25} with $\hat \phi = e^{i\omega t} \hat \psi$, and where  $\hat \psi$ is smooth up-to-boundary.
\qed
\end{Proposition}

Recall that in~\cite{AstefaneseiRadu} similar Einstein-scalar field solutions have been constructed with, however, $\mcV \equiv 0$. Our condition $\mcV'' <0$ is evaded there by letting $-\omega^2$  be an eigenvalue of the operator $\psi\mapsto VD_i(VD^i\psi)$. It would be of interest to provide a proof of existence of such solutions using our methods; compare~\cite{BizonWasserman} for the spherically symmetric asymptotically flat case. We plan to return to this in a near future.

The discussion of the finiteness of energy of the resulting field configurations, to be found in Section~\ref{s10I17.1}, is identical to the $\omega=0$ case.

\input{fofR}

\section{Stationary {metrics}}
 \label{ss18XI16.2+}
 \label{ss8XII16.3}

{We return to the metric \eq{gme1}-\eq{gme2}.
Let $e^{\hat \mu}$ be the coframe} $e ^{\hat 0}:=  dt + \theta$,  $e^{\hat i}:=dx^i $. The corresponding components  $R_{\hat \mu \hat \nu}$  of the Ricci tensor of $\glorentz $ read (see, e.g.,
\cite{Coquereaux:1988ne})
\begin{equation}\label{mainequation}
\left\{\begin{array}{l}
 R_{\hat 0 \hat 0}
  =  V \Delta_g V +\frac 1{4} |\newF|_g^2  \, ,\\
       R_{\hat i \hat j} = \Ric(g)_{ij}  -V^{-1}\Hess_gV _{ij} + \frac{1}{2V^{2}}\newF_{i  k} \newF_j{}^k\, ,
\\
 R_{\hat 0 \hat j}= -
  \frac 12 V^{-1} \nablariem _i (V \newF^ i{}_j)\, ,
\end{array}\right.
\end{equation}
where
$$
 \newF_{ij}=-V^2(\partial_i \theta_j - \partial_j
\theta_i)
 \, .
$$

For a general energy-momentum tensor $T$ the equations are
%
%
\begin{equation}\label{mainequations_T}
\left\{\begin{array}{l}
V(-\Delta_g V + n V) = \frac 1{4} |\newF|_g^2 - 8\pi G\left(T_{00}+\frac{V^2}{n-1}\Tr_{{}\glorentz }T\right)\, ,\\
\begin{split}
    R_{ij}+n g_{ij}-V^{-1}\nablariem _i\nablariem _j V= \frac{1}{2V^{2}} \lambda_{ik} \lambda^k{}_j
    + 8\pi G \big(& \theta_i \theta_j T_{00} - \theta_i T_{0j} - \theta_j T_{i0}\\
                           & +T_{ij}-\frac{g_{ij}}{n-1} \Tr_{{}\glorentz } T\big)\, ,
\end{split}
\\
\nablariem ^j (V \newF_{ij})= 16 \pi G V  (T_{0i} - \theta_i T_{00})\, .
\end{array}\right.
\end{equation}

When $T$ is given by \eqref{21XI16.1+}) we have
%
$$
 \Tr _{\glorentz } T =- \frac 1 {  8 \pi G} \left(   \frac {n-3} 2 W(\phi) |F|^2 +\frac {n-1} 4 |\nablariem  \phi|^2 + \frac {n+1}2 \mcV (\phi)
  \right)
 \,.
$$
%

Assuming moreover $\partial_t \phi =0$ we obtain
%
%
\begin{equation}\label{mainequations_EMD}
\left\{\begin{array}{l}
V(-\Delta_g V + n V) = \frac{1}{4} |\newF|_g^2 - 2 W F_{0i}F_0{}^i + \frac{V^2}{n-1}({\mycal V} - W |F|^2)\, ,\\
\begin{split}
    R_{ij}+n g_{ij}-V^{-1}\nablariem _i\nablariem _j V=&
     \frac{1}{2V^{2}} \lambda_{ik} \lambda^k{}_j
    + \frac{1}{2}(\partial_i \phi)(\partial_j \phi)+2W F_{i\alpha}F_j{}^\alpha
\\
    &+\frac{g_{ij}}{n-1}({\mycal V}-W |F|^2)
\\
    &- 2 W F_{0k}(F_j{}^k \theta_i + F_i{}^k \theta_j-F_0{}^k \theta_i \theta_j)\, ,
\end{split}
\\
\nablariem ^j (V \newF_{ij})= 4V W F_{0j} (F_i{}^j - F_0{}^j \theta_i)\, .
\end{array}\right.
\end{equation}
The matter equations {remain formally unchanged, as compared to \eq{22XI16.1}, when written in the form
\begin{equation}\label{matterequations}
\left\{\begin{array}{l}
\frac{1}{V\sqrt{\det g}}\partial_\mu (V \sqrt{\det g} W F^{\mu\nu})+B_{\mathrm{CS}}^\nu=0\,,\\
\frac{1}{V \sqrt{\det g}}\partial_i (V \sqrt{\det g} g^{ij} \partial_j \phi) - W'(\phi) |F|^2 - {\mycal V}'(\phi)=0\,,
\end{array}\right.
\end{equation}
with   $B_{\mathrm{CS}}^\nu$ as in \eq{8XII16.3}; indeed, the theta-dependent terms are hidden in $|F|^2$ and $F^{\mu\nu}$.
Letting $F=d(U dt+ A_i dx^i)$ and $\partial_t U = \partial_t A =0$ we have
\begin{equation}\label{matterequations2}
\left\{\begin{array}{l}
\nablariem _j (V W g^{jk}g^{il}(A_{l,k}-A_{k,l}+\theta_k U_{,l}-\theta_l U_{,k}))
 + V B_{\mathrm{CS}}^i=0\,,\\
\nablariem _j (V W g^{jk}(-V^{-2}U_{,k}+g^{lm}\theta_l(\theta_mU_{,k}-\theta_k U_{,m}
 +A_{k,m}-A_{m,k})))
 +V B_{\mathrm{CS}}^0=0\,,\\
V^{-1}\nablariem _i (V g^{ij} \partial_j \phi) - W'(\phi) |F|^2 - {\mycal V}'(\phi)=0\,,
\end{array}\right.
\end{equation}
where
\[
|F|^2 = 2\left[(A_{j,i}-A_{i,j})g^{ik}g^{jl}(A_{l,k}-2U_{,k}\theta_l)+|\nabla U|_g^2(|\theta|_g^2-V^{-2})-(U_{,i} \theta^i)^2\right]\,.
\]

Let $\epsilon >0$. For $\theta$ small in $C^{k,\alpha}_\epsilon$-norm the last two operators in \eqref{matterequations2}, acting on $\phi$ and $U$, are close in norm to the operators considered in Section~\ref{s11XII16.1}, and therefore isomorphisms as discussed there. Hence, in the implicit function argument we can  choose $\hat U$ and $\hat\phi$ as in Section~\ref{s11XII16.1}. It remains to consider the $A$-equation. We set
\beal{1XI16.2}
  \sigma_F
   &\equiv &
   {\nabla }_\mu A^\mu \equiv \frac 1 {\sqrt {\det \glorentz}}\partial_\mu (\sqrt {\det \glorentz} \glorentz^{\mu\nu} A_\nu)
\\
 \nn
  &=& V^{-1}D_i\left(V ( A^i - U \theta^i)\right)
   \,.
\eea
Then the first line of \eqref{matterequations2} can be rewritten as
\beal{1XII16.2}
\phantom{xxxxx}
\lefteqn{ \nablariem _j (V W g^{jk}g^{il}(\nablariem _kA_{l }-\nablariem _l A_{k}+\theta_k U_{,l}-\theta_l U_{,k}))+V B^i_\mathrm{CS}}
 &&
\\
 &=&    \nablariem _j (V W g^{jk}g^{il}(\nablariem _kA_{l }+\theta_k U_{,l}-\theta_l U_{,k}))
 \nn
\\
 &&
 -  \nablariem ^k (V W  ) \nablariem ^i A_{k} -   V W \nablariem _j  \nablariem ^i A^j
 + V B^i_\mathrm{CS}
\nn
\\
 &=&    \nablariem _j (V W g^{jk}g^{il}(\nablariem _kA_{l }+\theta_k U_{,l}-\theta_l U_{,k}))-  \nablariem ^k (V W  ) \nablariem ^i A_{k}
 \nn
\\
 &&
- VW R^i{}_\ell A^\ell -   V W   \nablariem ^i \left(\sigma_F -V^{-1} A^k \nablariem_ kV +
    V^{-1}\nablariem_k(U V\theta^k)\right)
  \nn
\\
 &&
+ V B^i_\mathrm{CS}
 \,.
  \nn
  \eea
%
%

The resulting  linear operator acting on $A$ coincides with the one in \eq{21XII16.2},
so that
the discussion there applies.
Inserting the asymptotic expansions for $V$ and $g$ into \eqref{1XI16.2} gives
\begin{equation}\label{21XII16.4}
\sigma_F=\nabla^\mu A_\mu=\rho^{1+n}\partial_\rho(\rho^{1-n}(A_\rho-\theta_\rho U))+O(\rho)
\end{equation}
and, as $U=O(1)$, setting $\hat A_\rho |_{\partial M} \equiv 0$ guarantees $\sigma_F\equiv 0$, as discussed at the end of Section~\ref{ss18XI16.2}.

The rest of the proof is an application of the implicit function theorem, we sketch the details.
We work with
$$(V,g)=(\mathring V+v,\mathring g+h)
$$
close to $(\mathring V,\mathring g)$.
Keeping in mind that
$$
 \mbox{$\lambda=-V^2d\theta$, $W(\phi)=1+O(\phi)$ and
${\mycal V}'(\phi)={\mycal V}''(0)\phi+O(\phi^2)$,}
$$
the system obtained after the addition of the $\Omega$-terms as in \eqref{8XII16.14}
is of the form:
\begin{equation}
 \label{matterequations2withtheta}
 \left\{\begin{array}{l}
 \mathfrak Q (v,h,\theta)- q_1[v,h,\theta,A,U,\phi]=0
 \,,
\\
 \mathcal P(A)-d\sigma_F-q_2[v,h,\theta,A,U,\phi]=0
 \,,
\\
V\nablariem _j (V^{-1}\nablariem^jU)-q_3[v,h,\theta,A,U,\phi]=0
 \,,
\\
 V^{-1}\nablariem _i (V \nablariem^i \phi)  - {\mycal V}''(0)\phi- q_4[ v,h,\theta,A,U,\phi]=0
 \,,
\end{array}\right.
\end{equation}
where the $q_i$'s are at least quadratic in their arguments and their first derivatives, and where $\mathfrak Q (v,h,\theta)$ corresponds to the operators $(l,L,\mathcal L)$ of~\cite[Corollaries~3.2 and 3.3]{ChDelayStationary}, which are the three components of the operator $\frac12\Delta_L+n$,  with $(l,L)$ and $\mathcal L$ being isomorphisms.
For $s \in \R$ we define, as in~\cite{ChDelayEM}, the operators
$$
\mathcal T_s=V^{-s}\nablariem_i(V^s\nablariem^i\,\cdot).
$$
We consider the modified system \eqref{matterequations2withtheta} with the Maxwell  gauge-term $d\sigma_F$ added:
\begin{equation}\label{matterequations2withthetajauge}
 \left\{\begin{array}{l}
 \mathfrak Q (v,h,\theta)-q_1[v,h,\theta,A,U,\phi]=0
 \,,
\\
 \mathcal P(A)-q_2[v,h,\theta,A,U,\phi]=0\,,
\\
 \mathcal T_{-1}(U)-q_3[v,h,\theta,A,U,\phi]=0\,,
\\
 (\mathcal T_{1} - {\mycal V}''(0))(\phi)- q_4[v,h, \theta,A,U,\phi]=0\,.
 \end{array}\right.
\end{equation}
A solution, close to zero, of the elliptic system \eq{matterequations2withthetajauge}, with prescribed behavior at infinity, can be constructed  in the following way: Let us define
$$
 {\mathcal X}=(v,h,\theta,U,A,\phi)
\,,
$$
we want to solve  an equation of the form
$$\mathcal F({\mathcal X})=0
 \,,
$$
with ${\mathcal X}\sim\hat{{X}}$ at large distance, for a prescribed small
$$
 \hat{{X}} = (\hat v ,\hat h,\hat \theta,\hat U, \hat A, \hat \phi)
$$
(some of the components vanishing if desired).%
We let ${\mathcal X}=\hat{{X}}+{X}$
and define $$\mathfrak F(\hat{{X}},{X}):=\mathcal F (\hat{{X}}+{X}).$$
We have $\mathfrak F(0,0)=0$, with the linearisation
$$D_{{X}}\mathfrak F(0,0)=\text{diag}
 \big(
\mathfrak Q ,
\mathcal P,
\mathcal T_{-1},
 (\mathcal T_{1}- {\mycal V}''(0))
 \big)
$$
being an isomorphism. By the implicit
function  theorem, for all $\hat{{X}}$ small, there exist
a small  $ {{X}}$, depending smoothly on $\hat{{X}}$, such that $\mathcal X$ is a solution of \eq{matterequations2withthetajauge}.

We have already explained why this solution solves the desired original equations.

For completeness we provide examples of functional spaces  where the preceding procedure involving $\mathcal F$
applies. Taking into account the  weights needed so that each  of  the operators involved is an isomorphism,
a natural space for ${X}$ is, without indicating the tensor character of the relevant bundles,
\bel{21XII16.3}
{\mathcal E}^{k+2}:=
C^{k+2,\alpha}_{-1+s}\times C^{k+2,\alpha}_{s}\times C^{k+2,\alpha}_{1+s}
\times C^{k+2,\alpha}_{1+s}
\times C^{k+2,\alpha}_{s}\times C^{k+2,\alpha}_{\sigma_{-}+s},
\ee
where $s$ is greater than and close to zero. The space for
$\hat{ {X}}|_{\partial M} $ is
$$
\rho^{-1}C^{k+2,\alpha}\times
\rho^{0} C^{k+2,\alpha}
\times
\rho^{0}C^{k+2,\alpha}\times
\rho^{0} C^{k+2,\alpha}
\times \rho^{0} C^{k+2,\alpha}\times \rho^{\sigma_{-}}\,C^{k+2,\alpha}\,.
$$
The tensor fields in $C^{k+2,\alpha}(\partial M)$, can then be extended away from  $\partial M$ to smooth tensor fields on $\bar M$ in any
convenient way, keeping in mind the conditions
\bel{12I17.1+}
 \hat A_\rho |_{\partial M} = \hat \theta_\rho |_{\partial M} = \hat h_{\rho i}  |_{\partial M} =0
 \,.
\ee

The reader can check that with the spaces chosen above, both  $\mathfrak F(\hat X,.)$ and  $D_X\mathfrak F(0,0)$  map
 $\mathcal E^{k+2}$ to
 $\mathcal E^{k}$.

We have thus proved:

\begin{theorem}
   \label{T27XII16.1}
Suppose that the Einstein metric $(S^1\times M, \mathring V^2 dt^2 + \mathring g)$ is non-degenerate, has no harmonic one-forms which are in $L^2$, and  $\mcV''(0)$
is not an $L^2$-eigenvalue of the operator $\Delta_{\mathring \griem}$.
Assume that
\bel{8XII16.+1}
 W(0)=1
 \,,
  \quad
  \mcV(0)=0
   =
  \mcV'(0)
  \,,
   \quad
   \mcV''(0)> -n^2/4
  \,,
\ee
\begin{enumerate}
  \item
and $
\mcV''(0) < 0
$
with $\hat \theta_a|_{\partial M} dx^a$, $\hat U|_{\partial M}$, $\hat A_a|_{\partial M} dx^a$, and  $\hat \phi|_{\partial M}$ which are sufficiently small smooth fields on $\partial  M$, or

  \item
$\hat \phi \equiv 0$ and $\hat \theta_a|_{\partial M} dx^a$, $\hat U|_{\partial M}$, and $\hat A_a|_{\partial M} dx^a$ are sufficiently small smooth fields  on $\partial M$,
\end{enumerate}

Then  there exist a solution of the Einstein-Maxwell-dilaton-scalar fields-Chern-Simons equations, or of the Yang-Mills-Higgs-Chern-Simons-dilaton equations with a trivial principal bundle, so that near $\partial M$ we have
\bel{27XII16.1}
 g\to_{\rho \to 0} \mathring g
 \,, \
 V \to_{\rho \to 0} \mathring V
 \,, \ U\to_{\rho \to 0} \hat U
 \,, \ A \to_{\rho \to 0} \hat A_a dx^a
 \,, \ \theta \to_{\rho \to 0} \hat \theta
 \,,
\ee
with all convergences in $\mathring g$-norm. The hypothesis of non-existence of harmonic $L^2$-one-forms is not needed if $\hat A_a|_{\partial M} dx^a \equiv 0 \equiv \hat U|_{\partial M} $, in which case the Maxwell field or the Yang-Mills field are identically zero.
\qed
\end{theorem}

\begin{remark}
 \label{R24III17.1}
Let us comment on the kernel conditions above.

First, we show in Appendix \ref{s29XII16.2} that the condition of non-existence of $L^2$-harmonic forms is satisfied near anti-de Sitter space-time in any case.

Next, it has been shown by Lee~\cite[Theorem A]{lee:spectrum} that there are no $L^2$-eigenvalues of $\Delta_{\mathring \griem}$
when the Yamabe invariant of the conformal infinity is
positive, in particular near anti-de Sitter space-time. Furthermore, and quite generally,  $\mcV''(0)=0$ is never an eigenvalue; this is, essentially, a consequence of the maximum principle. Finally, again quite generally, the $L^2$ spectrum of $-\Delta_{\griem}$ for asymptotically hyperbolic manifolds
is $[n^2/4,+\infty[$ together with possibly a finite set of eigenvalues, with finite multiplicity, between $0$ and $n^2/4$~\cite{Guillarmou} (compare~\cite{MazzeoMelrose}),
so our non-eigenvalue  condition is true except for at most a finite number of values of $\mcV''(0)\in (-n^2/4,0)$ for all asymptotic geometries.
%
\qed
\end{remark}

\begin{remark}\label{remkillvariable}
Some comments on properties of the solutions  are in order.

The case $ (\hat v|_{\partial M} ,\hat h_{ab}|_{\partial M}  dx^adx^b,\hat \theta_a|_{\partial M} dx^a) =0$  leads to a solution with the usual AdS conformal boundary when $\mathring \glorentz$ is taken to be the AdS metric.

The case  $ \hat \theta_a|_{\partial M} dx^a=\hat U |_{\partial M} = \hat A_a|_{\partial M} dx^a= \hat \phi |_{\partial M} =0$ with $|\hat v|_{\partial M} |+|\hat h_{ab}|_{\partial M}  dx^adx^b|\ne 0$ leads to the non-trivial vacuum configurations constructed in~\cite{ChDelayStationary}.

Note that $\hat \theta_a|_{\partial M} dx^a \equiv 0$  but $ \hat U \not \equiv 0$ and $ \hat A \not \equiv 0$ might lead to a solution with $  \theta\not \equiv 0$ because the off-diagonal terms of the Maxwell energy-momentum tensor will drive a non-zero $\theta$.

If we choose $\hat A_a|_{\partial M} dx^a \equiv 0$,   $\hat U|_{\partial M} \equiv 0$ and  $\hat \theta_a|_{\partial M} dx^a\equiv 0$ we obtain, from uniqueness of solutions,  static solutions with scalar fields.

The vanishing of $\phi|_{\partial M}$  will lead to $\phi= 0$ everywhere only if
$W'(0)=0$, since otherwise the equation for $\phi$ is non-homogeneous.
\qed
\end{remark}

The discussion of Section~\ref{ss5III17.1} applies in the stationary case without changes, so  that Theorem~\ref{T27XII16.1} also leads to solutions of $f(R)$ theories for the class of functions $f$ described there.

\subsection{Time-periodic scalar fields}
 \label{ss11I17.1}
 Consider  time-periodic scalar fields of the form
\begin{equation}
 \label{12I17.11}
  \phi(t,x)=e^{i\omega t}\psi (x)
 \,,
 \
 {\omega  \in \R}
 \,,
\end{equation}
as done in the static case in section \ref{ss19XII16.1}, but where now $\psi(x)$ is allowed to be complex. Assume for simplicity that all Maxwell fields are Abelian. Using the notation \eq{12I17.1}, and adapting the action as before gives for the Einstein equations
\begin{equation}\label{mainequations_EMD_periodicscalar}
\renewcommand{\arraystretch}{2}
\left\{\begin{array}{l}
V(-\Delta_g V + n V) = \frac{1}{4} |\newF|_g^2 -\frac{1}{2}\omega^2|\psi|^2 - 2 \GW F_{0i}F_0{}^i + \frac{V^2}{n-1}(\GmcV- \GW |F|^2)\, ,
\\
\begin{split}
    R_{ij}+n g_{ij}-& V^{-1}\nablariem _i\nablariem _j V=
     \frac{1}{2V^{2}} \lambda_{ik} \lambda^k{}_j
    + \frac{1}{2}\Re(\partial_i \psi \partial_j \bar{\psi})\\
    &+\frac{1}{2}(\theta_i\theta_j \omega^2 |\psi|^2 - \omega\theta_i \Im(\bar{\psi}\partial_j\psi)-\omega\theta_j\Im(\bar{\psi}\partial_i \psi))\\
    &+2\GW F_{i\alpha}F_j{}^\alpha+\frac{g_{ij}}{n-1}(\GmcV-\GW |F|^2)
\\
    &- 2 \GW F_{0k}(F_j{}^k \theta_i + F_i{}^k \theta_j-F_0{}^k \theta_i \theta_j)\, ,
\end{split}
\\
V^{-1}\nablariem ^j (V \newF_{ij})=4\GW F_{0j} (F_i{}^j - F_0{}^j\theta_i)+\omega\Im(\bar{\psi} \partial_i\psi)-\theta_i\omega^2 |\psi|^2 \, ,
\end{array}\right.
\end{equation}
where $\Im$ denotes taking the imaginary part.
The matter equations are then
\begin{equation}\label{matterequations2_periodicscalar}
\renewcommand{\arraystretch}{2}
\left\{\begin{array}{l}
\nablariem _j (V \GW g^{jk}g^{il}(A_{l,k}-A_{k,l}+\theta_k U_{,l}-\theta_l U_{,k}))
 + V B_{\mathrm{CS}}^i=0\,,\\
 \begin{split}
\nablariem _j (V \GW g^{jk}(-V^{-2}U_{,k}+g^{lm}\theta_l(\theta_mU_{,k}-\theta_k U_{,m}
 +A_{k,m}-A_{m,k})))&
\\
 +V B_{\mathrm{CS}}^0&=0\,,
 \end{split}\\
\begin{split}
&V^{-1}D_i (Vg^{ij}\partial_j \psi)
- \big( \GW'|F|^2
 + \GmcV'\big) \psi\\
&\qquad+(V^{-2}-\theta_k\theta^k) \omega^2 \psi+ i\omega (\theta^j \partial_j\psi+V^{-1}D_j(V\theta^j \psi))=0\,.
 \end{split}
\end{array}\right.
\end{equation}

It should be clear that all our previous arguments apply to this system of equations, leading to

\begin{Proposition}
  \label{P21XII16.3stationary} Assuming \eq{12I17.1}, the conclusions of Theorem~\ref{T27XII16.1} concerning the Einstein-Maxwell-Chern-Simons-dilaton-scalar fields
  hold for all sufficiently small $\omega \in \R$ and  $\hat \psi|_{\partial M}$ when $\phi$ takes the form \eq{12I17.11}
with $\hat \phi = e^{i\omega t} \hat \psi$,  where  $\hat \psi$ is smooth up-to-boundary.
\qed
\end{Proposition}

\section{Asymptotics and energy}
 \label{s10I17.1}

Whatever follows applies to the  static, or stationary, or time-periodic solutions constructed above.

All our solutions have a polyhomogeneous expansion at $\partial M$,  that is, expansions in terms of integer powers of $\ln\rho$ and of suitable powers of $\rho$, as determined by all indicial exponents.  It is straightforward but tedious to obtain a detailed description of the asymptotic behavior by inserting polyhomogeneous expansions in the equations and comparing coefficients.

We emphasise that non-integer indicial exponents $\sigma_\pm$ for the scalar fields (cf.~\eq{1XII16.4}) will introduce non-integer powers of $\rho$ in asymptotic expansions for small $\rho$ of all fields involved unless some miraculous cancellations occur. Logarithms of $\rho$ are expected in the expansion for generic solutions regardless of whether or not $\sigma_\pm$ are in $\mathbb Z$.

We have already described the leading-order behaviour of the matter fields. Concerning the metric, the following holds:
In the vacuum case, from what has been said it follows that the metrics under consideration have the following asymptotic behavior
\beal{28III17.1}
 &
 \rho  V-   \hat V=O(\rho^2)
  \,,
  \quad
  \theta_i- \hat\theta_i=O(\rho )
   \,,
 \quad
\rho^2g_{ij}- \hat{g}_{ij}=O(\rho^2)
 \,,
\eea
where all the fields at the left-hand sides have smooth limits at $\rho=0$ in local coordinates near the boundary, and where we have extended the fields $\hat V$, $\hat \theta$, and $\hat g$ from the conformal boundary to the interior by requiring them to be $\rho$-independent in some arbitrarily chosen coordinate system $(\rho,x^A)$ near the boundary.

In the presence  of matter, we  have to look at the error terms arising from the energy-momentum tensor.
On constant-$t$ slices the matter energy-density ${ T_{\hat 0 \hat 0} }\equiv V^{-2} T_{00}$ reads
\beal{21XI16.1++}
     {T_{\hat 0 \hat 0}} =\frac{1}{16\pi G}\big[ \big(\frac{1}{2}(V^{-2}\partial_0 \phi \partial_0 \phi +|d\phi|_g^2 )+  {\mycal V}
 + W (4 V^{-2}
           F_{0\mu}F_0{}^\mu
 +|F|^2)\big)\big]\,.
 \nn
\eea
 As in~\cite{ChDelayEM}, the total energy content of  non-trivial Maxwell {or Yang-Mills} fields will be finite only in space-dimensions $n=3$ and $n=4$.  If $\hat \phi|_{\partial M} \not \equiv 0$ the $\phi$-contribution to $\rho$ behaves as $\rho^{2\sigma_-}$, which will lead to a finite total energy of the scalar field if and only if
 \bel{21XII16.41}
  -n^2 < 4 \mcV''(0) < -n^2 +1
  \,.
 \ee
If $\hat \phi |_{\partial M} \equiv 0$  then either $\phi \equiv 0$ when there is no Maxwell field and so the space-time is vacuum, or there is a Maxwell field in which case the $\phi$-contribution to the energy is
\bel{18XII16.11+}
 \left\{
         \begin{array}{ll}
           O(\rho^{2\sigma_+}) & \hbox{if $\sigma_+< 4$;} \\
           O( \rho^{8}\ln^2 \rho) & \hbox{if $\sigma_+=4$,}
         \end{array}
       \right.
\ee
which gives a finite integral in either case.\\

 This behavior of the energy-momentum tensor  will affect the asymptotics \eq{28III17.1} only if  1) $0<\sigma_-<1$, in which case the relative
corrections
 of the metric components as above will be of order  $O(\rho^{2\sigma_-})$
in place of $O(\rho^2)$,
e.g.
$$
 V = \rho^{-1}\hat V( 1 + O(\rho^{2\sigma_-}))
  \,,
$$
 etc.;
or
if 2) $\sigma_-=1$, which would lead to relative corrections $O(\rho^{2\sigma_-}\ln \rho)$.

We note that the requirement of a well-defined Hamiltonian mass of the metric reads~\cite{ChNagy}, taking the previous constraint $0<\sigma_-<1$ into account,
\bel{31III17.1}
\min({2, 2\sigma_-}) >   n/2
 \,,
\ee
which is satisfied by our solutions with non-trivial scalar field with $\hat \phi\rho^{\sigma_-}$ asymptotic behaviour and with $\hat \phi \not \equiv 0$ if and only if
\bel{31III17.2}
 -n^2 < 4 \mcV''  (0) < - \frac{3 n^2} 4
 \,,
\ee
with $n=3$. On the other hand, all solutions with $\hat \phi \equiv 0$ have well defined and finite Hamiltonian mass.

The reader will note that, in all our solutions, only the scalar fields can possibly affect the leading-order asymptotics \eq{28III17.1} of the metric.

\bigskip

\appendix

\input{Omega}

\section{{The Christoffels}}
 \label{sA4I17.1}

For further reference,  and to get insight into \eq{8XII16.15},  we compute the Christoffels symbols of a stationary metric.
 The computations will be  made for a Lorentzian metric but the
change  $V\rightarrow iV$, where $i^2=-1$ gives the results for a Riemannian metric.

We recall the form of the metric considered here:
\beal{gme1a} &\glorentz  = -V^2(dt+\underbrace{\theta_i
dx^i}_{=\theta})^2 + \underbrace{g_{ij}dx^i dx^j}_{=g}\, , & \\ &
\partial_t V = \partial_t \theta = \partial_t g=0\, .
\eeal{gme2a}
Its inverse is
\beal{1XI16.1a}
 \glorentz^\sharp
  & \equiv
   &\glorentz^{\mu\nu}\partial_\mu\partial_\nu
   =
    -(V^{-2} - g^{ij}\theta_i\theta_j) \partial_0^2
   - 2g^{ij}\theta_i \partial_0 \partial_j
  +g^{ij} \partial_i\partial_j
\\
  & = &    - V^{-2}  \partial_0^2
   +g^{ij} (\partial_i- \theta_i\partial_0)(\partial_j- \theta_j\partial_0)
    \,.  \nn
\eea
The Christoffel symbols are then
%
%
%
$$
^\glorentz\Gamma^0_{00}=-\theta^iVD_iV,\, \, \,
^\glorentz\Gamma^k_{00}=VD^kV,
$$
$$
^\glorentz\Gamma^0_{0j}=-(|\theta|^2-V^{-2})VD_jV
+\frac12\theta^k[\partial_j(V^2\theta_k)-\partial_k(V^2\theta_j)],
$$
$$
^\glorentz\Gamma^0_{ij}=-\frac12(|\theta|^2-V^{-2})
[\partial_j(V^2\theta_i)+\partial_i(V^2\theta_j)]
-\theta^l\Gamma_{lij}(g-V^2\theta\otimes\theta)
 \,,
$$
$$
^\glorentz\Gamma^k_{i0}=\theta^kVD_iV
-\frac12 g^{kj}[\partial_i(V^2\theta_j)-\partial_j(V^2\theta_i)]
 \,,
$$
$$
^\glorentz\Gamma^k_{ij}=
\frac12 \theta^{k}[\partial_i(V^2\theta_j)+\partial_j(V^2\theta_i)]+g^{kl}\Gamma_{lij}(g-V^2\theta\otimes\theta)
 \,,
$$
where, as usual, $g^{ij}$ is inverse to $g_{ij}$ and where
$$
\Gamma_{ijk}=\frac{1}{2}(g_{ik,j}+g_{ji,k}-g_{jk,i})\,.
$$

\section{AdS, $\H^n$ and a quotient}
 \label{s29XII16.2}

Let us consider the Poincar\'e ball model of the hyperbolic space $\H^n$.
The hyperbolic space is then the unit ball of $\R^n$ endowed with the hyperbolic metric
$$
\mathring g=\rho^{-2}\delta \, ,
$$
where $\delta $ is the Euclidean metric and
$$
\rho(x)=\frac12(1-|x|^2) \, .
$$
We define
$$
\mathring V=\rho^{-1}-1 \, .
$$
A model for $(n+1)$-dimensional hyperbolic space is
$$
\R\times \H^n \, ,
$$
equipped  with the warped product metric
$$
\mathring\griem=\mathring V^2dt^2+\mathring g \, ,
$$
we thus write as usual
$$
\H^{n+1}=\R\, \,  _{\mathring V^2}\hspace{-1mm}\times\,  \H^n \, .
$$
Anti-de Sitter space is the same manifold with the Lorentzian metric
$$
\mathring\griem=-\mathring V^2dt^2+\mathring g \, .
$$
Let $\Gamma=\mathbb Z \subset \R$ be a discrete subgroup of isometries of the $\R$ factor of $\H^{n+1}$. Then  we can write
$$\Gamma\backslash \H^{n+1}=\mathbb S^1\, \,  _{\mathring V^2}\hspace{-1mm}\times\,  \H^n \, .
$$
Recall that the limit set $\Lambda(\Gamma)$ (see eg.~\cite{Ratcliffe2006},  §12.1 page 573) is a subset of the sphere at infinity of $\H^{n+1}$
consisting of the union of the limits of all the orbits. {We will show that  $\Lambda(\Gamma)$ consists of two points at infinity. It then follows from~\cite[Theorem C]{CarronPedon} that our quotient has no $L^2$ harmonic one-forms.

To justify our claim about $\Lambda(\Gamma)$, consider} the half space model of $\H^{n+1}$ with $(\vec x, y)\in\R^n\times(0,\infty)$
endowed with the metric $ds^2=y^{-2}(|d\vec x|\,^2+dy^2)$. Let $\mathcal H$ be the half sphere
$|\vec x|^2+y^2=1$ with $y>0$. It is well known that the totally geodesic hypersurface $\mathcal H$ with the induced metric $h$ is a model
of the hyperbolic space $\H^n$.
We reparameterise $\H^{n+1}$ with $(\vec u,\sqrt{1-|\vec u|^2})\in\mathcal H$  and $t\in \R$ by
$$
(\vec x, y)=e^t(\vec u,\sqrt{1-|\vec u|^2}).
$$
In the coordinate system $(t,\vec u)\in\R\times B^n(0,1)$, the metric becomes
$$ds^2=(1-|\vec u|^2)^{-1}dt^2+h.$$ Comparing the Ricci curvature of the Einstein metrics $ds^2$ and $h$ (see appendix of~\cite{ChDelayEM}),
we see that the radial function $V= (1-|\vec u|^2)^{-1/2}$ on $\mathcal H$ satisfies Hess$_h\,V=V\,h$. As we also have  min$\,V=1$,
the space $\H^{n+1}=(\R\times\mathcal H, ds^2)$ is then $(\R\times\H^n,\mathring\griem)$.

Now for $\vec u $  fixed with $t\rightarrow\pm\infty$ the point $(\vec x, y)$
tends to $0$ or infinity. These correspond to two {antipodal} points of the sphere at infinity in the Poincar\'e
ball model.

\section{An isomorphism on functions}
 \label{sA16III17.1}

Let $s$ and $\lambda$ be real numbers. We will need an isomorphism property for the following operator  acting on functions:
$$
\sigma\mapsto ({\mathcal
T}_s+\lambda)\sigma:=V^{-s}\nabla^i(V^{s}\nabla_i\sigma)+\lambda\sigma=
\nabla^i\nabla_i\sigma+sV^{-1}\nabla^iV\nabla_i\sigma+\lambda \sigma\,.
$$

We note that, whenever $V^2d\varphi^2+g$ is a static asymptotically
hyperbolic metric on $\mbbS^1\times M$, it holds that
\bel{30XI16.1}
 \mbox{$V^{-2}|dV|^2\to 1$ and $V^{-1}\nabla^i\nabla_iV\to n$}
\ee
as the conformal boundary is approached.

\begin{theorem}\label{Tisofunction}
{
Let $(M,g)$ be an $n$-dimensional Riemannian manifold with an asymptotically hyperbolic metric with $V>0$  and assume that \eq{30XI16.1} holds.
}
%
Let $s\neq 1-n$ and  $\lambda<(\frac{s+n-1}2)^2$ suppose that
 $$
 \left|\delta- \frac{s+n-1}2\right|<\sqrt{\left(\frac{s+n-1}2\right)^2-\lambda}
 \,.
$$
If $T_s+\lambda:=V^{s/2}({\mathcal T}_s+\lambda)V^{-s/2}$ has no $L^2$-kernel, 
then $\mathcal T_s+\lambda$ is an isomorphism from
$C^{k+2,\alpha}_{\delta}( M)$ to $C^{k,\alpha}_{\delta}(M)$.
\end{theorem}

\begin{proof}
The proof is an easy adaptation of the one of Theorem 3.3 in~\cite{ChDelayEM}.
We sketch the steps.
We set $\sigma=V^{-\frac{s}{2}}f$, thus
\bel{ZcalZ}
 {\mathcal
T}_s\sigma=V^{{-}\frac{s}{2}}\left[\nabla^i\nabla_if-\frac s2\left((\frac s2-1)
V^{-2}|dV|^2+V^{-1}\nabla^i\nabla_iV\right)f\right]=:
 V^{{-}\frac{s}{2}}T_sf\,.
\ee
{By assumption} $V^{-2}|dV|^2\to 1$ and
$V^{-1}\nabla^i\nabla_iV\to n$ at the conformal boundary,
leading to the following indicial exponents  {for $T_s+\lambda$:}
$$
  \delta=\frac{n-1}2\pm \sqrt{\left(\frac{s+n-1}2\right)^2-\lambda}\,.
$$
The calculation immediately after Lemma 3.4 of~\cite{ChDelayEM} shows that the
operator $T_s+\lambda$ satisfies condition (1.4) of
\cite{Lee:fredholm},
\bel{Lec1.4}\|u\|_{L^2}\le C \|({
T}_s+\lambda)u\|_{L^2}
 \,,
\ee
for smooth $u$ compactly supported in a
sufficiently small open set ${\mathcal U}\subset M$ such that
$\overline{\mathcal U}$ is a neighborhood of $\partial M$, with
$$
 C^{-1}=\frac{(s+n-1)^2}{4}-\lambda\,.
$$
We conclude using~\cite{Lee:fredholm}, Theorem C(c),
keeping in mind that $V^{-\frac s2}f$ is in $C^{k,\alpha}_{\delta}(M)$
iff $f\in C^{k,\alpha}_{\delta-\frac s2}(M)$, and our hypothesis
that there is no $L^2$-kernel.
\end{proof}

\noindent{\sc Acknowledgements} The research of PTC was supported in
part by the Austrian Research Fund (FWF), Project  P 24170-N16. PK is supported by a uni:docs grant of the University of Vienna.

\bibliographystyle{amsplain}

\bibliography{%
../references/newbiblio,%
../references/reffile,%
../references/bibl,%
../references/prop,%
../references/hip_bib,%
../references/newbib,%
../references/PDE,%
../references/netbiblio,%
stationaryEMD-minimal}
\end{document}

%% file: method.tex
\section{Method}
 \label{ss18XI16.2}

We seek to construct Lorentzian metrics ${}\glorentz $ in any
space-dimension $n\geq 3$, with Killing vector
$X=\partial/\partial t$. In adapted coordinates those metrics can
be written as
\beal{gme1} &\glorentz  = -V^2(dt+\underbrace{\theta_i
dx^i}_{=\theta})^2 + \underbrace{g_{ij}dx^i dx^j}_{=g}\,, & \\ &
\partial_t V = \partial_t \theta = \partial_t g=0\,.
\eeal{gme2}
Let us denote by $\varphi=(\varphi^a)$ all matter fields, where the index $a$ runs over some index set $\{1,\ldots,N_m\}$, $1\le N_m < \infty$. The $\varphi^a$'s will be required to satisfy
\bel{8XII16.11}
 \partial_t \varphi = 0
\ee
in the coordinate system of \eq{gme1}, except in Section~\ref{ss19XII16.1}, respectively Section~\ref{ss11I17.1},
 where time-periodic matter field configurations are considered with static, respectively stationary, metrics. In the case of the action \eq{21XI16.2} we thus have $\varphi=(A_\mu, \phi^a)$, but the overall argument applies to more general systems as long as the energy-momentum tensor is at least quadratic in the fields.
%

Consider the Einstein equations \eq{11X16.1}, in space-time dimension $n+1$, with a cosmological constant $\Lambda$. We impose \eq{gme1}-\eq{8XII16.11}, and assume that the energy-momentum tensor $T$ does not depend upon more than one derivative of $\glorentz$. We further suppose that
\bel{8XII16.12}
 \mbox{
\emph{whenever the matter field equations are satisfied we have $\nabla_\mu T^\mu{}_\nu = 0$,}}
\ee
regardless of whether or not the metric $\glorentz$ satisfies \eq{11X16.1}. (We use the symbol $\nabla$ to denote the covariant derivative of the metric $\glorentz$.)  The following approach has become standard since~\cite{ChBActa} in this kind of problems, we review the method for completeness. We assume that the fields $\varphi$ satisfy a system of equations of the form
\bel{8XII16.13}
 L( \glorentz,\partial\glorentz,\partial^2\glorentz)[\varphi]=0
 \,,
\ee
where $L$ is a partial-differential operator acting on $\varphi$ with coefficients which might depend upon $\glorentz$ and its derivatives up to order two.

In order to obtain an elliptic system of equations for $(V,\theta,g)$ we replace \eq{11X16.1} by
\bel{8XII16.14}
    R(\glorentz )_{\mu\nu} -\frac{R( \glorentz )} 2 {} \glorentz_{\mu\nu}  +\Lambda \glorentz_{\mu\nu} + \nabla_\mu \Omega_\nu + \nabla_\nu \Omega_\mu
    - {\nabla^\alpha \Omega_\alpha \glorentz _{\mu\nu}}
    = 8 \pi G T_{\mu\nu}
 \,,
\ee
where
\begin{eqnarray}
 \label{8XII16.15}
 \Omega_\nu
&:=&
\glorentz_{\nu\mu}\glorentz^{\alpha \beta} (\Gamma(\mathring{\glorentz})^\mu_{\alpha\beta}-\Gamma(\glorentz)^\mu_{\alpha\beta}  )
\end{eqnarray}
(cf.\ Appendix~\ref{sA4I17.1} below).
Assume that \eq{8XII16.13}-\eq{8XII16.14} can be solved for $(V,\theta,g,\varphi)$. The Bianchi identity and \eq{8XII16.12} imply that
\bel{8XII16.16}
  \nabla^\mu(\nabla_\mu \Omega_\nu + \nabla_\nu \Omega_\mu
    - {\nabla^\alpha \Omega_\alpha \glorentz _{\mu\nu}} ) = 0
 \,.
\ee
We show in Appendix~\ref{A29XII16.1} that this equation implies  {$\Omega  \equiv 0$ whenever $|\Omega|_\griem=o(\rho^{-1})$} (as will be the case for our solutions), and consequently we will obtain the desired solution of the original equations.

For the Einstein-scalar field equations this is the end of the story, provided we can construct solutions of  the system \eq{8XII16.13}-\eq{8XII16.14}. This will be done using the implicit function theorem around the solution $\phi\equiv 0$, $\glorentz = \mathring \glorentz$. Now, the associated linearised equations are rather complicated, in particular the question of invertibility of the linearisation of the modified Einstein equations is not a trivial issue. This is solved in~\cite{ACD,ACD2,ChDelayStationary,ChDelayEM} by the following artefact, which we apply again  here: It is well known that the Einstein tensor for a Riemannian metric
$$
 \griem = V^2 dt^2 + g
$$
on $S^1\times M$ coincides with the Einstein tensor of $\glorentz = -V^2 dt^2 + g $. This implies that the isomorphism property of the linearised operator for the {``harmonically reduced Riemannian Einstein equations''}, at a static solution, carries over to the Lorentzian equations; {compare~\cite[Section~3 and Appendix~A]{ChDelayStationary}.}
Hence, our hypothesis of non-degeneracy of $\griem$ together with the implicit function theorem can be used to obtain solutions of the Lorentzian equations, provided that suitable isomorphism theorems can be established for the matter equations.  This will be done for the equations at hand in Sections~\ref{ss18XI16.1} and \ref{ss18XI16.2+} below.

In cases involving Maxwell fields there will arise an issue related to gauge freedom for Maxwell fields  which will be addressed in a somewhat similar manner to the addition of the $\Omega$-terms to the Einstein tensor: To render \eq{8XII16.13} well posed we will add to it a ``gauge-fixing'' term $\sigma_F$, which will have to be shown to vanish.
For definiteness we consider the equations resulting from \eq{21XI16.2}:
\beal{10XII16.1}
  \lefteqn{
  \nabla^\mu (W(\phi)(\nabla_\nu A_\mu - \nabla_\mu A_\nu))
  }
 &&
\\
 \nn
 &&
 \equiv -\nabla^\mu (W \nabla_\mu A_\nu )-W R^\alpha{}_\nu A_\alpha + W' \nabla^\mu \phi \nabla_\nu A_\mu + W \nabla_\nu
    \underbrace{\nabla^\mu A_\mu }_{=: \sigma_F}
\\
 &&
= -\lambda \epsilon_{\nu \mu_1 \ldots \mu_{2k}}
  F^{\mu_1 \mu_2} \ldots F^{\mu_{2k-1}\mu_{2k}}
 \,,
 \nn
\eea
and note that the divergence of the rightmost term above vanishes.
We will show that we can solve the equation obtained by setting $\sigma_F$ to zero in \eq{10XII16.1}:
\beal{10XII16.2}
  &&-\nabla^\mu (W \nabla_\mu A_\nu )-W R^\alpha{}_\nu A_\alpha + W' \nabla^\mu \phi \nabla_\nu A_\mu
\\
 &&
 \phantom{xxxxxxx}
    = -\lambda \epsilon_{\nu \mu_1 \ldots \mu_{2k}}
  F^{\mu_1 \mu_2} \ldots F^{\mu_{2k-1}\mu_{2k}}
 \,.
 \nn
\eea
Equivalently,
\beal{10XII16.3}
  \lefteqn{
  \nabla^\mu (W(\phi)(\nabla_\nu A_\mu - \nabla_\mu A_\nu))
  }
 &&
\\
 &&
=  - W \nabla_\nu
    \nabla^\mu A_\mu -\lambda \epsilon_{\nu \mu_1 \ldots \mu_{2k}}
  F^{\mu_1 \mu_2} \ldots F^{\mu_{2k-1}\mu_{2k}}
 \,.
 \nn
\eea
Since the divergence of the left-hand side vanishes, we obtain
\beal{10XII16.4}
 \nabla ^\nu( W \nabla_\nu
    \underbrace{ \nabla^\mu A_\mu}_{=:\sigma_F} ) =0
 \,.
\eea
It follows e.g.\ from~\cite[Theorem~3.3]{ChDelayEM} that if $\nabla^\mu A_\mu \to_{\rho\to0} 0$, then $\nabla^\mu A_\mu \equiv 0$, so that the solution of \eq{10XII16.3} solves \eq{10XII16.1}. This implies that the energy-momentum tensor is divergence-free, and we conclude as in the case without Maxwell fields.

{For some purposes it is convenient to replace \eq{8XII16.14} with its equivalent form
\bel{31XII16.2}
 R_{\alpha\beta} =  T_{\alpha\beta}
 -
 \nabla_\alpha \Omega_\beta - \nabla_\beta \Omega_\alpha
 + \frac{2\Lambda  - \glorentz^{\mu\nu} T_{\mu\nu}}{n-1} \glorentz_{\alpha\beta}
 \,.
\ee
We note that the linearisation at $\mathring\glorentz$ of the $\Omega$-contribution above is
\bel{31XII16.32}
\mathring\glorentz^{\mu\nu} (
 \mathring \nabla_\alpha    \mathring \nabla_\mu h _{\beta \nu } +
  \mathring \nabla_\beta     \mathring \nabla_\mu h _{\alpha \nu })
   -
   \mathring \nabla_\alpha  \mathring \nabla_\beta (\mathring\glorentz^{\mu\nu}  h _{\mu \nu })
 \,,
\ee
which cancels exactly the non-$\Delta_L$ terms in \eq{31XII16.11}.

%% file: fofR.tex
\subsection{$f(R)$ theories}
 \label{ss5III17.1}

It is well known that $f(R)$ theories can be reduced to the Einstein-scalar field equations with a specific potential. The object of this section is to show that there exists a class of functions $f$ for which our method applies, providing static or stationary asymptotically AdS solutions.

More precisely, we consider  $f(R)$-vacuum theories as described in~\cite[Section~2.3]{DeFelice2010} (compare~\cite[Equation~32]{KijowskiJakubiecNonlinearPRD}), except for interchanging $g$ and $\tilde g$. We assume that the space-dimension $n$ equals three.
There one starts with the action
\[
S=\frac{1}{2\kappa^2}\int d^4x \sqrt{-\tilde{g}} f(\tilde{R})\,,
\]
and defines $F:=f'(\tilde R)$, assuming $F>0$.
As reported in~\cite{DeFelice2010}, a conformal transformation $g:= F \tilde{g}$,  brings the action to the Einstein-scalar field form
\begin{equation}
S_E=\int d^4 x \sqrt{-g}\left[\frac{1}{2\kappa^2} R -\frac{1}{2} g^{\mu\nu}\partial_\mu\phi\partial_\nu \phi-{\mycal V}(\phi)\right]
 \,,
\end{equation}
where
\[
\phi :=\frac{\sqrt{3/2}}{\kappa}\ln F\,,
\quad {\mycal V}(\phi)=\frac{F \tilde{R} - f(\tilde{R})}{2\kappa^2 F^2}\,.
\]
The question then arises, whether there exist functions $f$ which satisfy the hypotheses on $
\mcV$  set forth in Section~\ref{ss2I16.1}.

To ensure that we can invert $f'$ to obtain $\tilde{R}=f'^{-1}(F)$ we assume $f'' \neq 0$.

Using the definitions of $F$ and $\phi$ gives
\begin{equation}
{\mycal V}(\phi)=\frac{e^{-2\kappa\sqrt{\frac{2}{3}}\phi}}{2\kappa^2}
\left(
	e^{\sqrt{\frac{2}{3}}\kappa \phi}f'^{-1}(e^{\sqrt{\frac 2 3}\kappa\phi})
	- f(f'^{-1}(e^{\sqrt{\frac 2 3} \kappa \phi} ))
\right)\,.
\end{equation}
The first and second derivatives of ${\mycal V}(\phi)$ are given by
\begin{align}
{\mycal V}'(\phi) &= \frac{\sqrt{2/3}}{\kappa} \left(
	e^{-2\sqrt{\frac 2 3}\kappa \phi}f(f'^{-1}(e^{\sqrt{\frac 2 3} \kappa \phi}))
	- \frac 1 2 e^{-\sqrt{\frac 2 3}\kappa \phi}f'^{-1}(e^{\sqrt{\frac 2 3} \kappa \phi})
\right)\,,\\
{\mycal V}''(\phi) &= \frac{1}{3}\Big(
	e^{-\sqrt{\frac 2 3} \kappa \phi} f'^{-1}(e^{\sqrt{\frac 2 3} \kappa \phi})
	- 4 e^{-2 \sqrt{\frac 2 3} \kappa \phi}f(f'^{-1}(e^{\sqrt{\frac 2 3} \kappa \phi}))\\\nonumber
	&\hspace{3em}
	+ \underbrace{(f^{\prime -1})^{ \prime}(e^{\sqrt{\frac 2 3}\kappa \phi})}_{\hspace{-3em}=1/(f''(f'^{-1}(\exp(\sqrt{2/3}\kappa \phi))))\hspace{-3em}}
\Big)\,.
\end{align}
Evaluating  at $\phi=0$ gives
\begin{align*}
{\mycal V}(0) &= \frac{1}{2 \kappa^2}\left(
	f'^{-1}(1) - f(f'^{-1}(1))
\right)\,,\\
{\mycal V}'(0) &= \frac{\sqrt{2/3}}{\kappa}\left(
	f(f'^{-1}(1)) - \frac 1 2 f'^{-1}(1)
\right)\,,\\
{\mycal V}''(0) &= \frac 1 3 \left(
	 f'^{-1}(1) - 4 f(f'^{-1}(1)) + \frac{1}{f''(f^{\prime -1}(1))}
\right)\,.
\end{align*}

Set  $a:=f'^{-1}(1)$.
To obtain a negative cosmological constant in the action for $g$ we need ${\mycal V}(0)<0$ and therefore $f(a)>a$.
The condition
${\mycal V}'(0)=0$ leads to $f(a)=a/2$, hence $a<0$.

Summarising, we require
\begin{equation}
f'>0\,,\quad f''\neq 0
 \,,\quad  f'^{-1}(1)<0
  \,,\quad f(f'^{-1}(1))=f'^{-1}(1)/2
 \,.
\end{equation}

To provide an example of function $f$ which satisfies our requirements, we let
$$
 f(\tilde{R})=d \tilde{R}+c \tilde{R}^{\alpha+1} + e
  \,,
$$
with constants $c$, $d$ and $e$. This gives
\be
F =d+c(\alpha+1)\tilde{R}^\alpha
 \quad
 \Longleftrightarrow
 \quad
 \left(\tilde{R}(F)\right)^\alpha =\frac{F-d}{c(\alpha+1)}\,.
\ee
We consider $d<1, c<0, \alpha=1, 3, 5,\dots,$ and
\[
e=\frac{\alpha(1-2d)-1}{2(\alpha+1)} \sqrt[\alpha]{\frac{1-d}{c(\alpha+1)}}\,,
\]
where $\sqrt[\alpha]{.}$ denotes the real root.
Then $a=\sqrt[\alpha]{\frac{1-d}{c(\alpha+1)}}<0$, ${\mycal V}'(0)=0$ and ${\mycal V}(0)<0$.

Additionally our method requires ${\mycal V}''(0)>-n^2/4=9/4$ (recall that $n=3$ in~\cite{DeFelice2010})
 and we need ${\mycal V}''(0)<0$ to obtain nontrivial solutions for $\phi$ without Maxwell fields. These conditions are fulfilled if
\begin{equation}\label{3II17.1}
-\sqrt[\alpha]{\frac{1-d}{c(\alpha+1)}}+\frac{\left(\sqrt[\alpha]{\frac{1-d}{c (\alpha+1)}}\right)^{1-\alpha}}{c \alpha (\alpha+1)}+\frac{27}{4}
 \equiv\frac{a}{(1-d)\alpha}-a+\frac{27}{4}>0
\end{equation}
and $d>(\alpha-1)/\alpha$. Equation~\eqref{3II17.1} can be satisfied by choosing either $c$ large enough (as the l.h.s is of the form $K/\sqrt[\alpha]{c}+27/4$ for some constant $K$ independent of $c$), or $\alpha$ large enough (as it is of the form $\sqrt[\alpha]{K_1/(\alpha+1)}(1-1/(\alpha K_2))+27/4$ for   positive constants  $K_{1,2}$ independent of $\alpha$).

%% file: Omega.tex
\section{The $\Omega$ equation}
 \label{A29XII16.1}

 {Consider  \eq{8XII16.15}:
\bel{8XII16.16++}
(^\glorentz P \Omega)_{\nu}:=  \nabla^\mu(\nabla_\mu \Omega_\nu + \nabla_\nu \Omega_\mu
  -  \nabla^\alpha \Omega_\alpha \glorentz _{\mu\nu}
    ) = 0
 \,.
\ee
}
We have
\beal{27XII16.21}
 \lefteqn{-(\,^{\glorentz}P\Omega)^\nu=
 \nabla_{\mu }(\nabla^\mu \Omega^\nu + \nabla^\nu \Omega^\mu)
   {-  \nabla^\nu \nabla^\alpha \Omega_\alpha }
 }
&&
\\
 &= &
\nn
  \nabla_{\mu }(\nabla^\mu \Omega^\nu - \nabla^\nu \Omega^\mu+ 2 \nabla^\nu \Omega^\mu)
   {-  \nabla^\nu \nabla^\alpha \Omega_\alpha }
\\
 & = &
  \nabla_{\mu }(\nabla^\mu \Omega^\nu - \nabla^\nu \Omega^\mu)+   \nabla^\nu \nabla_{\mu }\Omega^\mu
 + 2 R^\nu{}_\mu \Omega^\mu
 \nn
\\
 & =  &
  \frac 1 {\sqrt{|\det \glorentz |}} \partial_{\mu }\big(\sqrt{|\det \glorentz|}\glorentz^{\mu\alpha}\glorentz^{\nu\beta}
    (\partial_\alpha \Omega_\beta - \partial_\beta \Omega_\alpha)\big)
    \nn
\\
\nn
 &&
    +   \glorentz^{\nu\alpha}\partial_\alpha\Big[\frac 1 {\sqrt{|\det \glorentz |}}\partial_\mu\big( \sqrt{|\det \glorentz|}
        \glorentz^{\mu\beta}\Omega_\beta\big)
        \Big]
 + 2 R^\nu{}_\mu \Omega^\mu
 \,.
\eea
In the static case, with a time-independent $\Omega_\nu$ and an Einstein $(n+1)$-dimensional metric $\glorentz$ normalised so that $R_{\mu\nu} = -n \glorentz_{\mu\nu}$,
 this reads
\beal{28XII16.1}
 \lefteqn{
 \nabla_{\mu }(\nabla^\mu \Omega_0 + \nabla_0 \Omega^\mu)
  = \glorentz_{00}
 \nabla_{\mu }(\nabla^\mu \Omega^0 + \nabla^0 \Omega^\mu)
 }
&&
\\
 \nn
 & =  & \glorentz_{00}
  \frac 1 {\sqrt{|\det \glorentz |}} \partial_{i }\big(\sqrt{|\det \glorentz|}\glorentz^{i j}\glorentz^{00}
    \partial_j \Omega_0 \big)
 - 2 n  \Omega_0
\\
 \nn
 & =  &
  \frac V { \sqrt{\det g }} \partial_{i }\big(V^{-1} \sqrt{\det g}g^{i j}
    \partial_j \Omega_0 \big)
 - 2 n  \Omega_0
\\
 \nn
 & =  &
  V D_{i }(V^{-1} D^i \Omega_0 )
 - 2 n \Omega_0
\,,
\\
 \lefteqn{
 \nabla_{\mu }(\nabla^\mu \Omega_k + \nabla_k \Omega^\mu)
  =
 }
&&
\\
 & =  &
  \frac 1 {\sqrt{|\det \glorentz |}} \glorentz _{k \ell} \partial_{i }\big(\sqrt{|\det \glorentz|}\glorentz^{i j}\glorentz^{\ell m}
    (\partial_j \Omega_m - \partial_m \Omega_j)\big)
    \nn
\\
\nn
 &&
    +   \partial_k \Big[\frac 1 {\sqrt{|\det \glorentz |}}\partial_i\big( \sqrt{|\det \glorentz|}
        \glorentz^{ij}\Omega_j\big)
        \Big]
 - 2 n  \Omega_k
 \,.
\eea
One notices that  $^\glorentz P$ coincides with its Riemannian analogue $^\griem P$ when mapping covectors to covectors:
$$
 ^\glorentz P = ^\griem \!\!  P
 \,.
$$
In particular if
\bel{28XII16.3}
 \nabla_\mu ( \nabla^\mu \Omega^\nu + \nabla^\nu \Omega^\mu
  {-  \glorentz^{\mu\nu} \nabla^\alpha \Omega_\alpha }) =0
\ee
in the Lorentzian metric, then the same
remains true in the Riemannian one.

Consider a one-form $\Omega$ which is in $L^2$ and satisfies
{
$$^\griem P \Omega_\nu  \equiv \nabla_\mu \nabla^\mu \Omega_\nu + R^\mu{}_\nu \Omega_\mu =0
 \,.
$$
}
Multiplying \eq{28XII16.3} by $-\griem ^{\nu\mu}\Omega_\mu$, and integrating by parts over $S^1\times M$ \emph{using $\griem$ everywhere} one finds
\beal{28XII16.4}
 \phantom{xxx}
 0 = \int_{S^1\times M} \nabla_\mu \Omega_\nu \nabla^\mu \Omega^\nu
  + n \Omega^\alpha \Omega_\alpha
  \equiv
 \int_{S^1\times M} |\nabla \Omega|^2_\griem + n | \Omega |^2_\griem
 \,,
\eea
where the vanishing of the boundary term follows from elliptic regularity and density arguments. Thus $ \Omega _{\mu }\equiv 0$; equivalently, the $L^2$ kernel of {$^\griem P$} is trivial.

As $^\griem P$ is an elliptic, formally
self-adjoint operator, geometric in the sense of~\cite{Lee:fredholm},  it is  an isomorphism from $C^{k+2,\alpha}_\delta$ to $C^{k,\alpha}_\delta$ by the results in this last reference
when
$$
    \left|\delta-\frac {n} 2\right|<  \frac{n }2 + 1
 \,.
$$
Equivalently, in local coordinates, the indicial exponents   belong to $\{-2,n\}$.
In particular any $C^{0,\alpha}_{-1+\epsilon}$  one-form, with $\epsilon>0$, in the kernel of $ \,^{\griem}P $
is trivial.  In our applications, we will have $\Omega$ in the kernel with $\Omega\in C^{k+1,\alpha}_{s}$
where  $s$ is positive close to zero when $(\hat v,\hat h)=0$ and $s=0$ if
$(\hat v,\hat h)\neq0$,  which guarantees the vanishing of $\Omega$.

The usual perturbation arguments show that the above conclusions will remain true for all sufficiently small $\theta$.

The alert reader will have noted that the above discussion goes through for all negative-definite Ricci tensors such that $R_{0i}=0$.

\bigskip